\documentclass[twoside,10pt]{article} 
\usepackage{amsfonts}
\usepackage{amsmath}
\usepackage{amssymb}
\usepackage{epsfig}
\usepackage{graphicx}
\usepackage{latexsym}
\usepackage{subfigure}
\usepackage{times}
\usepackage{hyperref}

\newcommand{\app}{\textsc{Sharon}}
\newcommand{\greta}{{\small GRETA}}

\newcommand{\return}{\textsf{\footnotesize RETURN}}
\newcommand{\pattern}{\textsf{\footnotesize PATTERN}}

\newcommand{\where}{\textsf{\footnotesize WHERE}}
\newcommand{\group}{\textsf{\footnotesize GROUP-BY}}
\newcommand{\within}{\textsf{\footnotesize WITHIN}}
\newcommand{\slide}{\textsf{\footnotesize SLIDE}}

\newcommand{\mycount}{\textsf{\footnotesize COUNT}}
\newcommand{\mymin}{\textsf{\footnotesize MIN}}
\newcommand{\mymax}{\textsf{\footnotesize MAX}}
\newcommand{\mysum}{\textsf{\footnotesize SUM}}
\newcommand{\myavg}{\textsf{\footnotesize AVG}}

\newcommand{\mystart}{\textsf{\footnotesize START}}
\newcommand{\mymid}{\textsf{\footnotesize MID}}
\newcommand{\myend}{\textsf{\footnotesize END}}

\makeatletter
\providecommand*{\cupdot}{%
  \mathbin{%
    \mathpalette\@cupdot{}%
  }%
}
\newcommand*{\@cupdot}[2]{%
  \ooalign{%
    $\m@th#1\cup$\cr
    \hidewidth$\m@th#1\cdot$\hidewidth
  }%
}

\usepackage[ruled]{algorithm}
\PassOptionsToPackage{noend}{algpseudocode}
\usepackage{algpseudocode}

\renewcommand{\algorithmiccomment}[1]{\bgroup\hfill//~#1\egroup}

\algnewcommand\algorithmicswitch{\textbf{switch}}
\algnewcommand\algorithmiccase{\textbf{case}}
\algnewcommand\algorithmicassert{\texttt{assert}}
\algnewcommand\Assert[1]{\State \algorithmicassert(#1)}%
\algdef{SE}[SWITCH]{Switch}{EndSwitch}[1]{\algorithmicswitch\ #1\ \algorithmicdo}{\algorithmicend\ \algorithmicswitch}%
\algdef{SE}[CASE]{Case}{EndCase}[1]{\algorithmiccase\ #1}{\algorithmicend\ \algorithmiccase}%
\algtext*{EndSwitch}%
\algtext*{EndCase}%

\usepackage[font=small,labelfont=bf,tableposition=top]{caption}

\usepackage{listings}
\renewcommand{\algorithmiccomment}[1]{/* #1 */}

\usepackage{amsthm}
\usepackage{setspace,xspace}
\graphicspath{{figures/}}

\newcommand{\nop}[1]{}

\newcommand{\rem}[1]{\marginpar{\flushleft{#1}}}
\renewcommand{\rem}[1]{} 

\renewcommand{\algorithmiccomment}[1]{/* #1 */}

\graphicspath{{figures/}}
 \newtheorem{lemma}{Lemma}
\newtheorem{definition}{Definition}
\newtheorem{example}{Example}

\usepackage{multirow}
\newcommand{\eat}[1] {}

\usepackage{fancyhdr}

\fancypagestyle{plain}{%
  \fancyhead{}
  
}

 \newlength{\hoehe}
 \newlength{\breite}
\title{\fontsize{15}{15}\selectfont Sharon: Shared Online Event Sequence Aggregation\\\vspace*{1cm}
\large Technical Report\\
January 15, 2018
\vspace*{1cm}}
\author{\large Olga Poppe$^*$, Allison Rozet$^*$, Chuan Lei$^{**}$, Elke A. Rundensteiner$^*$, and David Maier$^{***}$}
\date{\Large 
\large $^*$Worcester Polytechnic Institute, Worcester, MA 01609\\
$^{**}$IBM Research, Almaden, 650 Harry Rd, San Jose, CA 95120\\
$^{***}$Portland State University, 1825 SW Broadway, Portland, OR 97201\\
*opoppe$|$amrozet$|$rundenst@wpi.edu, **chuan.lei@ibm.com, ***maier@pdx.edu
\vfill
}

\begin{document}
\maketitle

\begin{spacing}{0.8}
{\footnotesize \noindent \textbf{Copyright} \copyright{} 2018 by
Olga Poppe. Permission to make digital or hard copies of all or
part of this work for personal use is granted without fee provided
that copies bear this notice and the full citation on the first
page. To copy otherwise, to republish, to post on servers or to
redistribute to lists, requires prior specific permission. }
\end{spacing}

\clearpage
\pagestyle{fancy}

\clearpage
\tableofcontents

\pagenumbering{arabic}
\setcounter{page}{1}

%
%

\newpage
\begin{abstract}
Streaming systems evaluate massive workloads of event sequence aggregation queries. State-of-the-art approaches suffer from long delays caused by not sharing intermediate results of similar queries and by constructing event sequences prior to their aggregation. To overcome these limitations, our Shared Online Event Sequence Aggregation (\app) approach shares intermediate aggregates among multiple queries while avoiding the expensive construction of event sequences. Our \app\ optimizer faces two challenges. One, a sharing decision is not always beneficial. Two, a sharing decision may exclude other sharing opportunities. To guide our \app\ optimizer, we compactly encode sharing candidates, their benefits, and conflicts among candidates into the \app\ graph. Based on the graph, we map our problem of finding an optimal sharing plan to the Maximum Weight Independent Set (MWIS) problem. We then use the guaranteed weight of a greedy algorithm for the MWIS problem to prune the search of our sharing plan finder without sacrificing its optimality. The \app\ optimizer is shown to produce sharing plans that achieve up to an 18-fold speed-up compared to state-of-the-art approaches.
\end{abstract}
\section{Introduction}
\label{sec:introduction}

Complex Event Processing (CEP) is a prominent technology for supporting time-critical streaming applications ranging from public transport to e-commerce. CEP systems continuously evaluate massive query workloads against high-rate event streams to detect event sequences of interest, such as vehicle trajectories and purchase patterns. Aggregation functions are applied to these sequences to provide  summarized insights, such as the number of trips on a certain route to predict traffic jams or the number of items purchased after buying another item for targeted advertisement. CEP applications must react to critical changes of these aggregates in real time to compute best routes or increase profit~\cite{ADGI08, WDR06, ZDI14}.

\textbf{Motivating Examples}.
We now describe two time-critical multi-query event sequence aggregation scenarios. 

$\bullet$ \textit{\textbf{Urban transportation services}}. With the growing popularity of ridesharing services such as Uber and Lyft, their systems face multiple challenges including real-time analysis of vehicle trajectories, geospatial prediction, and alerting. These systems evaluate query workloads against high-rate streams of drivers' position reports and riders' requests to infer the current supply and demand situation on each route. They incorporate traffic conditions to compute the best route for each trip. They instantaneously react to critical changes to prevent waste of time, reduce costs and pollution, and increase riders' satisfaction and drivers' profit. With thousands of drivers and over 150 requests per minute in New York City~\cite{uber1,uber3}, real-time traffic analytics and ride management is a challenging task.

Queries $q_1$--$q_7$ in Figure~\ref{fig:queries} compute the count of trips on a route as a measure of route popularity. They consume a stream of vehicle-position reports. Each report carries a time stamp in seconds, a car identifier and its position. Here, event type corresponds to a vehicle position. For example, a vehicle on Main Street sends a position report of type \textit{MainSt}.
Each trip corresponds to a sequence of position reports from the same vehicle (as required by the predicate \textit{[vehicle]}) during a 10-minutes long time window that slides every minute. The predicates and window parameters of $q_2$--$q_7$ are identical to $q_1$ and thus are not shown for compactness.
Table~\ref{tab:candidates} summarizes the common patterns in this workload. For example, pattern $p_1=(\mathit{OakSt},$ $\mathit{MainSt})$ appears in queries $q_1$--$q_4$. Sharing the aggregation of common patterns among multiple similar queries is vital to speed up system responsiveness. 

\begin{figure}[t]
\centering
\includegraphics[width=0.6\columnwidth]{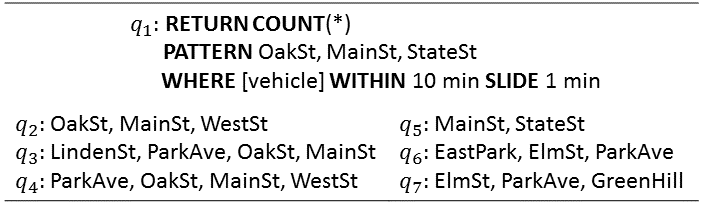}
\caption{Traffic monitoring workload $Q$}
\label{fig:queries}
\end{figure}
\begin{table}[t]
\centering
\begin{tabular}{|l|l|}
\hline
\textbf{Pattern} $p$ & \textbf{Queries} $Q_p \subseteq Q$ \\
\hline
\hline
$p_1=(\mathit{OakSt},\mathit{MainSt})$ & $q_1,q_2,q_3,q_4$ \\
\hline
$p_2=(\mathit{ParkAve},\mathit{OakSt})$ & $q_3,q_4$ \\
\hline
$p_3=(\mathit{ParkAve},\mathit{OakSt},\mathit{MainSt})$ & $q_3,q_4$ \\
\hline
$p_4=(\mathit{MainSt},\mathit{WestSt})$ & $q_2,q_4$ \\
\hline
$p_5=(\mathit{OakSt},\mathit{MainSt},\mathit{WestSt})$ & $q_2,q_4$ \\
\hline
$p_6=(\mathit{MainSt},\mathit{StateSt})$ & $q_1,q_5$ \\
\hline
$p_7=(\mathit{ElmSt},\mathit{ParkAve})$ & $q_6,q_7$ \\
\hline
\end{tabular}
\vspace{2mm}
\caption{Sharing candidates of the form $(p,Q_p)$ in the workload~$Q$}
\label{tab:candidates}
\end{table}

$\bullet$ \textit{\textbf{E-commerce}} systems monitor customer click patterns to identify the purchase of which item increases the chance of buying another item. Such purchase dependencies between items serve as a foundation for prediction, planning, and targeted ads on an online shopping website. 

Queries $q_8$--$q_{11}$ consume a stream of item purchases. Each event carries a time stamp in seconds, customer identifier, type of item, e.g., \textit{Laptop}. 
The choice of a laptop often determines the laptop cases, adapters, keyboard protectors, etc. that may be purchased next ($q_8$--$q_{10}$). A laptop may even determine a customer's phone preferences, e.g., Mac users are likely to choose an iPhone over other phones. The model of an iPhone, in turn, determines screen protectors for it ($q_{11}$).
Thus, queries $q_8$--$q_{11}$ compute the count of item sequences during a 20-minute time window that slides every minute. 
The pattern \textit{(Laptop, Case)} appears in all four queries in this workload. The aggregation of such patterns could be shared among these queries to achieve prompt updates of the recommendation model according to dynamically changing user preferences.

\begin{figure}[t]
\centering
\includegraphics[width=0.45\columnwidth]{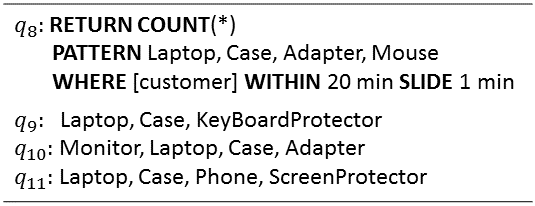}
\caption{Purchase monitoring workload}
\end{figure}

\textbf{State-of-the-Art Approaches} can be divided into three groups (Figure~\ref{fig:spectrum}):

$\bullet$ \textbf{\textit{Non-shared two-step approaches}}, including Flink~\cite{flink}, SASE~\cite{ZDI14}, Cayuga~\cite{DGPRSW07}, and ZStream~\cite{MM09}, evaluate each query independently from other queries in the workload. Furthermore, these approaches do not offer optimization strategies specific to event sequence aggregation queries. Without special optimization techniques, these approaches first construct event sequences and then aggregate them. Since the number of event sequences is polynomial in the number of events~\cite{ZDI14, QCRR14}, event sequence construction is an expensive step. Our experiments in Section~\ref{sec:evaluation} confirm that such a \textit{non-shared two-step} approach implemented by the popular open-source streaming system Flink~\cite{flink} does not terminate, even for low-rate streams of a few hundred events per second. 

\begin{figure}[t]
\centering
\includegraphics[width=0.6\columnwidth]{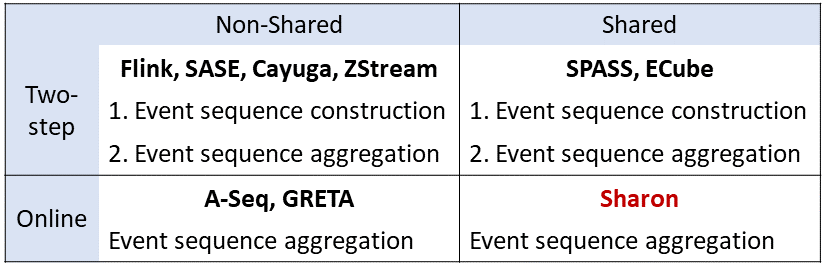}
\caption{Event sequence aggregation approaches}
\label{fig:spectrum}
\end{figure}

$\bullet$ \textbf{\textit{Shared two-step approaches}} such as SPASS~\cite{RLR16} and ECube~\cite{LRGGWAM11} focus on \textit{shared} event sequence construction, not on event sequence aggregation. If these approaches are applied to aggregate event sequences, they would construct all sequences prior to their aggregation. This event sequence construction step degrades system performance. Our experiments in Section~\ref{sec:evaluation} confirm that such a \textit{shared two-step} approach implemented by SPASS~\cite{ZDI14} requires 41 minutes per window, even for low-rate streams of a few hundred events per second. Such long delays are not acceptable for time-critical applications that require a response within a few seconds~\cite{linear_road}.

$\bullet$ \textbf{\textit{Non-shared online approaches}} such as A-Seq~\cite{QCRR14} and \greta~\cite{PLRM18} compute event sequence aggregation \textit{online}, i.e., without constructing the sequences. 
A-Seq incrementally maintains a set of aggregates for each pattern and instantaneously discards each event once it updates the aggregates.
\greta\ extends A-Seq to a broader class of queries and thus has higher computation costs. Neither of these approaches tackles the sharing optimization problem to determine which queries should share the aggregation of which patterns such that the latency of a workload is minimized -- which is the focus of this paper. These approaches lack optimizers that can identify shared computations among multiple queries.

\textbf{Challenges}.
We tackle the following open problems:

$\bullet$ \textit{\textbf{Online yet shared event sequence aggregation}}.
These two optimization techniques cannot be simply combined because they impose contradictory constraints on the underlying execution strategy. 
For example, if query $q_4$ \textit{shares} the aggregation results of patterns $p_2$ and $p_4$ with other queries (Table~\ref{tab:candidates}), the aggregates for $p_2$ and $p_4$ must be combined to form the final results for $q_4$. To ensure correctness, this result combination must be aware of the temporal order between sequences matched by $p_2$ and $p_4$ and their expiration. To analyze these temporal relationships, event sequences must be constructed. This requirement contradicts the key idea of the \textit{online} approaches that avoid event sequence construction. 

$\bullet$ \textit{\textbf{Benefit of sharing}}.
Sharing the aggregation computation for a pattern $p$ by a set of queries $Q_p$ containing $p$ is not always beneficial, since this sharing may introduce considerable CPU overhead for combining shared intermediate aggregates to form the final results for each query in $Q_p$. Thus, an accurate sharing benefit model is required to assess the quality of a sharing plan.

$\bullet$ \textit{\textbf{Intractable sharing plan search space}}.
The search space for a high-quality sharing plan is exponential in the number of sharing candidates (Table~\ref{tab:candidates}). Since the event rate may fluctuate, the benefit of sharing a pattern may change over time. To achieve a high sharing benefit, the sharing plan may have to be dynamically adjusted. Hence, an effective yet efficient optimization algorithm for sharing plan selection is required.

\begin{figure}[t]
\centering
\includegraphics[width=0.65\columnwidth]{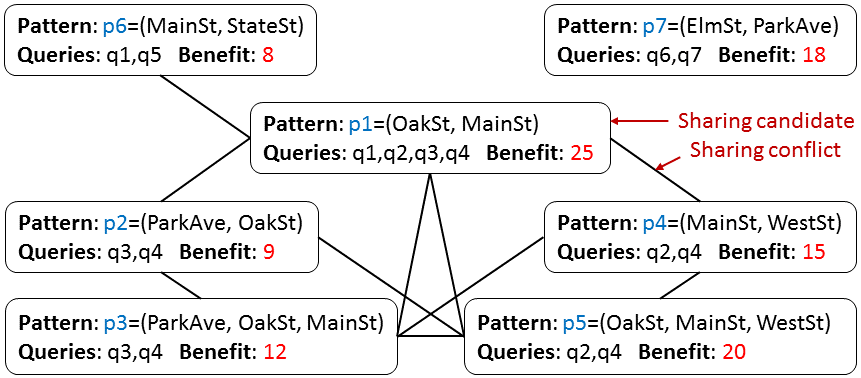}
\caption{\app\ graph for the traffic use case (Figure~\ref{fig:queries})}
\label{fig:graph}
\end{figure} 

\textbf{Our \app\ Approach}.
We propose the \underline{Shar}ed \underline{On}line Event Sequence Aggregation (\app) optimization techniques to tackle these challenges. 
Since sharing a pattern $p$ by a set of queries $Q_p$ is not always beneficial, we develop a \textit{sharing benefit model} to assess the quality of a sharing candidate $(p,Q_p)$. The model compares the gain of sharing $p$ among queries $Q_p$ to the overhead of combining shared aggregates of $p$ to form the final results for each query in $Q_p$. Non-beneficial candidates are pruned.
Since a decision to share a pattern may prevent the sharing of another pattern by the same query, we define the notion of \textit{sharing conflicts} among sharing candidates.
We compactly encode sharing candidates as vertices and conflicts among these candidates as edges of the \textit{\app\ graph} (Figure~\ref{fig:graph}). Each vertex is assigned a weight that corresponds to the benefit of sharing the respective candidate.
Based on the graph, we map our Multi-query Event Sequence Aggregation problem to the Maximum Weight Independent Set (MWIS) problem. We then utilize the guaranteed minimal weight of the approximate algorithm GWMIN~\cite{Sakai2003} for MWIS to prune conflict-ridden candidates. Since conflict-free candidates always belong to an optimal sharing plan, they can also be excluded from the search early on.
Based on the reduced graph, our sharing plan finder further prunes sharing plans with conflicts and returns an optimal plan (i.e., plan with minimal estimated latency) to guide our executor at runtime.
In summary, \app\ seamlessly combines two optimization strategies into one integrated solution. Namely, it \textit{shares} sequence aggregation among multiple queries, while computing sequence aggregation \textit{online}.

\textbf{Contributions}. 
Our key innovations are the following.

1)~We design the \textit{sharing benefit model} to assess the quality of a sharing candidate. Non-beneficial candidates are pruned.

2)~We identify \textit{sharing conflicts} among candidates and encode candidates, their benefits, and conflicts among them into the \textit{\app\ graph}.  

3)~We map our Multi-query Event Sequence Aggregation problem to the Maximum Weight Independent Set (MWIS) problem and utilize the guaranteed weight of the approximate algorithm for MWIS to prune conflict-ridden candidates.

4)~Based on the reduced \app\ graph, we introduce the \textit{sharing plan finder} that prunes sharing plans with conflicts and returns an optimal sharing plan.

5)~Our experiments using real data sets~\cite{uber1, linear_road} demonstrate that sharing plans produced by the \app\ optimizer achieve an 18-fold speed-up and use two orders of magnitude less memory compared to Flink~\cite{flink}, A-Seq~\cite{QCRR14}, and SPASS~\cite{RLR16}.

\textbf{Outline}. 
Section~\ref{sec:overview} provides an overview of our approach. 
We define the sharing benefit model and sharing conflicts in Sections~\ref{sec:benefit}--\ref{sec:graph}. 
We reduce the search space and design the sharing plan finder in Sections~\ref{sec:pruning}--\ref{sec:finder}. 
We discuss extensions of our core approach in Section~\ref{sec:discussion}.
Experiments are described in Section~\ref{sec:evaluation}. 
Section~\ref{sec:related_work} covers related work,
while Section~\ref{sec:conclusions} concludes the paper and describes future work.

\section{Sharon Approach Overview}
\label{sec:overview}

\subsection{Sharon Data and Query Model}
\label{sec:model}

\textit{Time} is represented by a linearly ordered \textit{set of time points} $(\mathbb{T},\leq)$, where $\mathbb{T} \subseteq \mathbb{Z^\geq}$ (non-negative integers).
%
An \textit{event} is a message indicating that something of interest to the application happened in the real world. An event $e$ has a \textit{time stamp} $e.time \in \mathbb{T}$ assigned by the event source. 
An event $e$ belongs to a particular \textit{event type} $E$, denoted $e.type=E$ and described by a \textit{schema} that specifies the set of \textit{event attributes} and the domains of their values. 
%
Events are sent by event producers (e.g., vehicles) on an input \textit{event stream I}. An event consumer (e.g., carpool system) monitors the stream with a workload of queries that detect and aggregate event sequences. We borrow the query syntax and semantics from SASE~\cite{ADGI08}. 

\begin{definition}(\textbf{Event Sequence Pattern}) 
Given event types $E_1, \dots E_l$, an \textbf{\textit{event sequence pattern}} has the form $P = (E_1 \dots E_l)$ where $l \in \mathbb{N},\ l \geq 1,$ is the length of $P$. 

Given a stream $I$, an \textbf{\textit{event sequence}} $s=(e_1 \dots e_l)$ is a match of a pattern $P$ if  
$\forall i,j \in \mathbb{N},$ $1 \leq i < j \leq l,$ the following conditions hold:
$e_i,e_j \in I$,
$e_i.\mathit{type} = E_i,$ $e_j.\mathit{type} = E_j$, and
$e_i.\mathit{time} < e_j.\mathit{time}$.
Event $e_1$ is called a \mystart\ event, $e_l$ is an \myend\ event, and $\{ e_2,\dots, e_{l-1} \}$ are \mymid\ events of $s$.
\label{def:pattern}
\end{definition}

\begin{definition}(\textbf{Event Sequence Aggregation Query}) 
An \textbf{\textit{event sequence aggregation query}} $q$ consists of five clauses:

$\bullet$ Aggregation result specification (\return\ clause),

$\bullet$ Event sequence pattern $P$ (\pattern\ clause),

$\bullet$ Predicates $\theta$ (optional \where\ clause),

$\bullet$ Grouping $G$ (optional \group\ clause), and

$\bullet$ Window $w$ (\within\ and \slide\ clause).

A query $q$ requires that all events in an event sequence $s$ are
matched by the pattern $P$ (Definition~\ref{def:pattern}),
satisfy the predicates $\theta$,
have the same values of the grouping attributes $G$, and 
are within one window $w$. 
Event sequences matched by $q$ are grouped by the values of $G$. The aggregation function of $q$ is applied to each group and a result is returned per group and per window. We focus on distributive (such as \mycount, \mymin, \mymax, \mysum) and algebraic aggregation functions (such as \myavg), since they can be computed incrementally~\cite{Gray97}. 

Let $e$ be an event of type $E$ and $\mathit{attr}$ be an attribute of $e$.
\mycount$(*)$ returns the number of all sequences per group, while
\mycount$(E)$ computes the number of all events $e$ in all sequences per group.
\mymin$(E.\mathit{attr})$ (\mymax$(E.\mathit{attr})$) returns the minimal (maximal) value of $\mathit{attr}$ for all events $e$ in all sequences per group.
\mysum$(E.\mathit{attr})$ calculates the summation of the value of $\mathit{attr}$ of all events $e$ in all sequences per group.
Lastly, \myavg$(E.\mathit{attr})$ is computed as \mysum$(E.\mathit{attr})$ divided by \mycount$(E)$ per group.
\label{def:query}
\end{definition}

\textbf{Assumptions}.
To initially focus the discussion, we assume that:
(1)~A pattern $p$ is shared among all queries containing $p$.
(2)~All queries have the same predicates, grouping, and windows.
(3)~An event type appears at most once in a pattern.
(4)~The workload is static.
We sketch extensions of our approach to relax these assumptions in Section~\ref{sec:discussion}.

\subsection{Sharon Framework}
\label{sec:framework}

We target \textit{time-critical} applications that require aggregation results within a few seconds (Section~\ref{sec:introduction}). 
%
%
Given a workload $Q$ and a stream $I$, the \textbf{\textit{Multi-query Event Sequence Aggregation (MESA) Problem}} is to determine which queries share the aggregation of which patterns (i.e., a sharing plan $\mathcal{P}$) such that the \textit{latency of evaluating the workload $Q$ according to the plan $\mathcal{P}$ against the stream $I$ is minimal}. 

To solve this problem, our \app\ framework deploys the following components (Figure~\ref{fig:overview}). 
Given a workload $Q$, our \textbf{Static Optimizer} finds an optimal sharing plan at compile time.
To this end, it identifies sharing candidates of the form $(p,Q_p)$ where $p$ is a pattern that could potentially be shared by a set of queries $Q_p \subseteq Q$.
It then estimates the benefit of each candidate $(p,Q_p)$    (Section~\ref{sec:benefit}),
determines sharing conflicts among these candidates, and 
compactly encodes all candidates, their benefits and conflicts into a \app\ graph (Section~\ref{sec:graph}).
Based on the graph, the optimizer prunes large portions of the search space  (Section~\ref{sec:pruning}) and returns an optimal sharing plan (Section~\ref{sec:finder}). 
Based on this plan, our \textbf{Runtime Executor} computes the aggregation results for each shared pattern and then combines these shared aggregations to obtain the final results for each query in the workload $Q$ (Section~\ref{sec:benefit}). 

\begin{figure}[t]
\centering
\includegraphics[width=0.7\columnwidth]{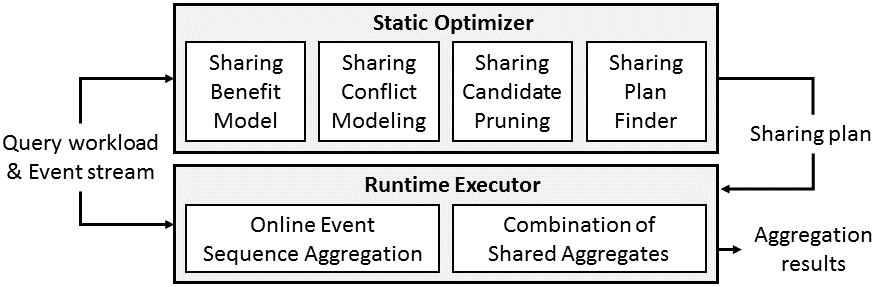}
\caption{\app\ framework}
\label{fig:overview}
\end{figure}

\section{Sharing Benefit Model}
\label{sec:benefit}

Our optimizer first identifies sharing candidates in a workload (Section~\ref{sec:candidate}).
It then decides whether to apply the \textit{Non-Shared} or the \textit{Shared method} to each query (Sections~\ref{sec:non-shared}--\ref{sec:shared}). Both methods are borrowed from A-Seq~\cite{QCRR14}.
Lastly, it estimates the benefit of sharing each candidate (Section~\ref{sec:bvalue}).

\subsection{Sharing Candidate}
\label{sec:candidate}

First, our optimizer identifies those patterns that could potentially be shared by queries in a given workload. 

\begin{definition}(\textbf{Sharable Pattern, Sharing Candidate})
Let $Q$ be a workload and $p$ be a pattern that appears in queries $Q_p \subseteq Q$. The pattern $p$ is \textit{sharable} in $Q$ if $p.length > 1$ and $|Q_p|>1$. A sharable pattern $p$ and queries $Q_p$ constitute a \textit{sharing candidate}, denoted as $(p,Q_p)$.
\label{def:sharable}
\end{definition}

Existing pattern mining approaches can detect sharable patterns. Due to space constraints, they are described in Appendix~\ref{app:ccspan}.

The pattern $P^i$ of a query $q_i \in Q_p$ consists of three sub-patterns, namely, $\mathit{prefix}^i$, $p$, and $\mathit{suffix}^i$ defined below. 


\begin{definition}(\textbf{Prefix and Suffix of a Sharable Pattern in a Query})
Let $P^i = (E^i_1 \dots E^i_n)$ be the pattern of a query $q_i$ and 
$p = (E^i_m \dots$ $E^i_{m+l})$ be a sharable pattern that appears in $q_i$
where $m,l,n \in \mathbb{N}$, $1 \leq m$, and $m+l \leq n$.
Then $\mathit{prefix}^i = (E^i_1 \dots E^i_{m-1})$ is called the \textit{prefix} and 
$\mathit{suffix}^i = (E^i_{m+l+1} \dots E^i_n)$ is called the \textit{suffix} of $p$ in $q_i$.
\label{def:sub-patterns}
\end{definition}

\subsection{Non-Shared Method}
\label{sec:non-shared}

While A-Seq~\cite{QCRR14} considers grouping, predicates, negation, and various aggregation functions, we now sketch only its key ideas, namely, online event sequence aggregation and event sequence expiration. We use event sequence count as an example, i.e., \mycount(*) (Definition~\ref{def:query}).

\textbf{Online Event Sequence Aggregation}
A-Seq computes the count of event sequences online, i.e., without constructing and storing these event sequences. To this end, it maintains a count for each prefix of a pattern. The count of a prefix of length $j$ is incrementally computed based on its previous value and the new value of the count of the prefix of length $j-1$.

\begin{example}
Let an event be described by its type and time stamp, e.g., $a_1$ is an event of type $A$ with time stamp 1. In Figure~\ref{fig:non-shared-1}, we count event sequences matched by the pattern $(A,B)$, denoted $\mathit{count}(A,B)$.
A count is maintained for each prefix of the pattern, i.e., for $(A)$ and $(A,B)$. 
The value of $\mathit{count}(A,B)$ is updated every time a $b$ arrives by summing $\mathit{count}(A)$ and $\mathit{count}(A,B)$. For example, when $b_4$ arrives, it is appended to each previously matched $a$ to form new sequences $(a_1,b_4)$ and $(a_2,b_4)$. The number of new sequences is $\mathit{count}(A)=2$. In addition, the previously formed sequence $(a_1,b_2)$ is kept. The number of previous sequences is $\mathit{count}(A,B)=1$. Thus, $\mathit{count}(A,B)$ is updated to 3.
\label{ex:non-shared-1}
\end{example}

\begin{figure}[t]
	\centering
    \subfigure[Online sequence aggregation]{
    	\includegraphics[width=0.3\columnwidth]{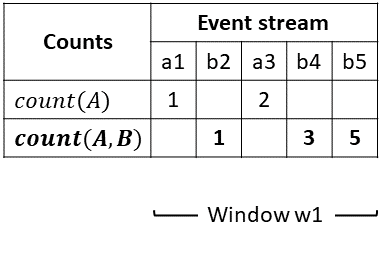}
    	 \label{fig:non-shared-1}
	}    
	\subfigure[Event sequence expiration]{
    	\includegraphics[width=0.3\columnwidth]{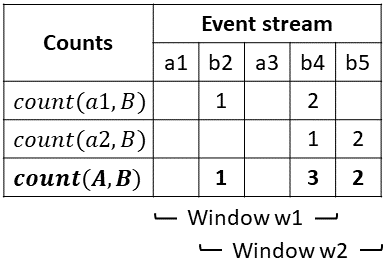}
    	\label{fig:non-shared-2}
	}
    \vspace*{-2mm}
	\caption{Non-Shared method}
	\label{fig:non-shared}
\end{figure}

\textbf{Event Sequence Expiration}.
Due to the sliding window semantics of our queries (Definition~\ref{def:query}), event sequences expire over time. To avoid the re-computation of all affected aggregates, we observe that a \mystart\ event of a sequence (Definition~\ref{def:pattern}) expires sooner than any other event in it. Thus, we maintain the aggregates per each matched \mystart\ event. When a new event arrives, only the counts of not-expired \mystart\ events are updated. When an \myend\ event $e$ arrives, it updates the final counts for all windows that $e$ falls into.

\begin{example}
In Figure~\ref{fig:non-shared-2}, assume a window of length four slides by one. A count is now maintained per each matched $a$. When $b_5$ arrives, $a_1$ is expired, $\mathit{count}(a_1,B)$ is ignored, $\mathit{count}(a_2,B)$ and $\mathit{count}(A,B)$ are updated for window $w_2$.
\label{ex:non-shared-2}
\end{example}

\textbf{Data Structures}.
Our \app\ Executor maintains 
a hash table that maps a pattern to its count.
Thus, we can access and update a count in constant time.

\textbf{Time Complexity}.
The query $q_i$ processes each event that it matches. The rate of matched events is computed as the sum of rates of all event types in the pattern $P^i$ of $q_i$ (Definition~\ref{def:sub-patterns}):
\begin{equation}
\begin{split}
\mathit{Rate}(P^i) = \sum_{j=1}^{n} \mathit{Rate}(E_j^i) 
\end{split}
\end{equation}

Since counts are maintained per \mystart\ event and an event type appears at most once in a pattern, each matched event updates one count per each not-expired \mystart\ event. There are $\mathit{Rate}(E_1^i)$ \mystart\ events. In summary, the time complexity of processing the query $q_i$ by the Non-Shared method is:
\begin{equation}
\begin{split}
\mathit{NonShared}(p,q_i) = \mathit{Rate}(E_1^i) \times \mathit{Rate}(P^i)
\label{eq:non-shared-1-query}
\end{split}
\end{equation}

For the set of queries $Q_p$, the time complexity corresponds to the summation of the time complexity for each query $q_i$.
\begin{equation}
\mathit{NonShared}(p,Q_p) = \sum_{q_i \in Q_p} \mathit{NonShared}(p,q_i)
\label{eq:non-shared-k-queries}
\end{equation}



\subsection{Shared Method}
\label{sec:shared}

Let $(p,Q_p)$ be a sharing candidate (Definition~\ref{def:sharable}).
Let $\mathit{preffix}^i$ and $\mathit{suffix}^i$ be the prefix and the suffix of $p$ in a query $q_i \in Q_p$ (Definition~\ref{def:sub-patterns}).
The aggregates for $\mathit{preffix}^i$, $p$, and $\mathit{suffix}^i$ are combined to obtain the aggregate for $q_i$. Due to event sequence semantics, the sequences matched by $\mathit{preffix}^i$ must appear before the sequences matched by $p$ which in turn must appear before the sequences matched by $\mathit{suffix}^i$ in the stream. To this end, the executor performs two steps:

(1)~\textbf{\textit{Count computation}}. 
Counts are maintained per each \mystart\ event of $\mathit{prefix}^i$, $p$, and $\mathit{suffix}^i$ (Section~\ref{sec:non-shared}).

(2)~\textbf{\textit{Count combination}}. 
The count of $\mathit{prefix}^i$ is multiplied with the count for each \mystart\ event of $p$. The resulting counts are summed to obtain $\mathit{count}(\mathit{preffix}^i,p)$. This count is further combined with the count of $\mathit{suffix}^i$ analogously.

\begin{figure}[t]
\centering
\includegraphics[width=0.42\columnwidth]{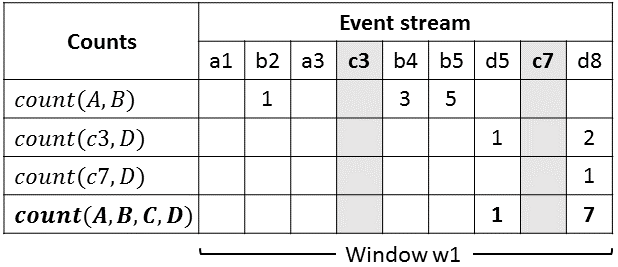}
\caption{Shared method}
\label{fig:shared}
\end{figure}

\begin{example}
In Figure~\ref{fig:shared}, we compute the count of $(A,B,$ $C,D)$ based on the counts of $(A,B)$ and $(C,D)$. Assuming that $a_1$--$d_8$ belong to the same window, $\mathit{count}(A,B)$ is computed as shown in Figure~\ref{fig:non-shared-1}. In addition, a count for each $c$ (i.e., $c_3$ and $c_7$) is maintained. 
When $c_3$ arrives, $\mathit{count}(A,B)=1$. We multiply it with $\mathit{count}(c_3,D)=2$ to obtain $\mathit{count}(A,B,c_3,D)=2$.
Analogously, when $c_7$ arrives, $\mathit{count}(A,B)=5$. It is multiplied with $\mathit{count}(c_7,D)=1$ to get $\mathit{count}(A,B,c_7,D)=5$.
Lastly, we sum these counts to obtain $\mathit{count}(A,B,C,D)=7$. 
\label{ex:executor}
\end{example}

\textbf{Time Complexity}.
1) \textbf{\textit{Count computation}}.  
Since the shared pattern $p$ is processed once for all queries in $Q_p$, the time complexity of processing each query $q_i$ by the Shared method corresponds to the sum of the time complexity of processing $\mathit{prefix}^i$ and $\mathit{suffix}^i$ of $q_i$.
\begin{equation} 
\begin{array}{lll}
\mathit{Comp}(p,q_i) &=& \mathit{Rate}(E_1^i) \times \mathit{Rate}(\mathit{prefix}^i) +\\
&& \mathit{Rate}(E_{m+l+1}^i) \times \mathit{Rate}(\mathit{suffix}^i)
\label{eq:comp}
\end{array}
\end{equation} 

2) \textbf{\textit{Count combination}}. 
The time complexity of count multiplication is the product of the number of counts. 
\begin{equation}
\begin{array}{lll}
\mathit{Comb}(p,q_i) &=& \mathit{Rate}(E_1^i) \times \mathit{Rate}(E_m) \times\\ && \mathit{Rate}(E_{m+l+1}^i)
\end{array}
\label{eq:comb}
\end{equation}

The time complexity of processing $q_i$ by the Shared method is the sum of the time complexity of these two steps.
\begin{equation}
\mathit{Shared}(p,q_i) = \mathit{Comp}(p,q_i) + \mathit{Comb}(p,q_i)
\label{eq:shared-1-query}
\end{equation}

For the set of queries $Q_p$, the time complexity corresponds to the summation of time complexity for each query $q_i$. In contrast to the Non-Shared method (Equation~\ref{eq:non-shared-k-queries}), the pattern $p$ is computed once by the Shared method.
\begin{equation}
\begin{array}{lll}
\mathit{Shared}(p,Q_p) &=& \mathit{Rate}(E_m) \times \mathit{Rate}(p) + \\
&& \sum_{q_i \in Q_p} \mathit{Shared}(p,q_i)
\end{array}
\label{eq:shared-k-queries}
\end{equation}


\subsection{Benefit of a Sharing Candidate}
\label{sec:bvalue}




\begin{definition}(\textbf{Benefit of a Sharing Candidate})
The benefit of sharing a pattern $p$ by the set of queries $Q_p$
is computed as the difference between the time complexity of the Non-Shared and Shared methods (Equations~\ref{eq:non-shared-k-queries} and \ref{eq:shared-k-queries}):
\begin{equation}
\begin{split}
\mathit{BValue}(p,Q_p) = \mathit{NonShared}(p,Q_p) - \mathit{Shared}(p,Q_p) 
\label{eq:bvalue}
\end{split}
\end{equation}
A sharing candidate $(p,Q_p)$ is called \textit{beneficial} if 
$\mathit{BValue}(p,$ $Q_p)>0$. It is called \textit{non-beneficial} otherwise.
\end{definition}

\textbf{\textit{Non-Beneficial Candidate Pruning}}.
All non-beneficial candidates can be excluded from further analysis.

Based on this cost model, we conclude that \textit{the number of queries, the length of their patterns}, and \textit{the stream rate} determine the benefit of sharing. We experimentally study the effects of these factors in Section~\ref{sec:evaluation}.

\section{Sharing Conflict Modeling}
\label{sec:graph}

A decision to share a pattern $p$ by a query $q \in Q_p$ may prevent sharing another pattern $p'$ by $q$ if $p$ and $p'$ overlap in $q$. Such sharing candidates are said to be in a \textit{sharing conflict}. We now encode sharing candidates, their benefit, and conflicts into the \app\ graph. Based on the graph, we then reduce the search space of our sharing plan finder (Sections~\ref{sec:pruning}--\ref{sec:finder}).

\begin{example}
In Table~\ref{tab:candidates}, queries $q_3$ and $q_4$ contain the overlapping patterns $p_2=(\mathit{ParkAve},\mathit{OakSt})$ and $p_1=(\mathit{OakSt},\mathit{MainSt})$.
Since the executor computes and stores the aggregates for a pattern as a whole (Section~\ref{sec:benefit}), $q_3$ and $q_4$ can either share $p_1$ or $p_2$, but not both. Thus, the sharing candidates $(p_1,\{q_1,q_2,q_3,q_4\})$ and $(p_2,\{q_3,q_4\})$ give ``contradictory instructions'' for $q_3$ and $q_4$. These candidates are said to be in a sharing conflict. 
However, if $p_1$ were to be shared only by $q_1$ and $q_2$, the sharing conflict between these candidates would be resolved. We sketch the techniques for sharing conflict resolution in Section~\ref{sec:discussion}. 
\label{ex:sharing-incompatibility}
\end{example}

\begin{definition}(\textbf{Sharing Conflict})
Let $p_A=(A_0\ldots A_n)$ and $p_B=(B_0\ldots B_m)$ be patterns and $Q_A$ and $Q_B$ be query sets. 
The sharing candidates $(p_A,Q_A)$ and $(p_B,Q_B)$ are in {\it sharing conflict} if $p_A$ overlaps with $p_B$ in a query $q \in Q_A \cap Q_B$, i.e.,
$\exists k \in \mathbb{N},$ $0 \leq k \leq n,m$ $A_{n-k}\ldots A_n = B_0\ldots B_k$ in $q$. 
The query $q$ causes the conflict between $(p_A,Q_A)$ and $(p_B,Q_B)$.
\label{def:sharing-conflict}
\end{definition}

\begin{definition}(\textbf{Valid Sharing Plan})
A \textit{sharing plan} $\mathcal{P}$ is a set of sharing candidates. $\mathcal{P}$ is \textit{valid} if it contains no candidates that are in conflict with each other. 
$\mathcal{P}$ is \textit{invalid} otherwise.
\label{def:valid}
\end{definition}

\begin{definition}(\textbf{Score of a Sharing Plan})
The \textit{score} of a sharing plan $\mathcal{P}=\{ (p_1,Q_{p1}), \ldots, (p_s,Q_{ps})\}$ is:
\begin{equation}
\mathit{Score}(\mathcal{P}) = \sum _{j=1} ^s {\mathit{BValue}(p_j, Q_{pj})}
\end{equation}
\label{def:score}
\end{definition}
\vspace*{-2mm}

\begin{definition}(\textbf{Optimal Sharing Plan})
Let $\mathbb{P}_{\mathit{val}}$ be the set of all valid sharing plans. $\mathcal{P}_{\mathit{opt}} \in \mathbb{P}_{\mathit{val}}$ is an {\it optimal} sharing plan if $\nexists \mathcal{P} \in \mathbb{P}_{\mathit{val}}$ with $\mathit{Score}(\mathcal{P}) > \mathit{Score}(\mathcal{P}_{\mathit{opt}})$.
\label{def:opt-sharing-plan}
\end{definition}

\begin{example}
Given the workload in Figure~\ref{fig:queries}, the plan $\mathcal{P}=\{(p_2,\{q_3,q_4\}); (p_4,\{q_2,q_4\})\}$ is valid. Its sharing candidates are not in conflict since the patterns $p_2=(\mathit{ParkAve},\mathit{OakSt})$ and $p_4=(\mathit{MainSt},\mathit{WestSt})$ do not overlap.
However, $\mathcal{P}$ is not an optimal plan because $\mathit{Score}(\mathcal{P})=24$ is not maximal among all valid plans. 
Indeed, another valid plan $\{(p_1,\{q_1,q_2,q_3,q_4\})\}$ has higher score 25.
\end{example}

\begin{definition}(\textbf{\app\ Graph})
Let $S$ be the set of sharable patterns in a workload $Q$. The \textit{\app\ graph} $G=$ $(V,E)$ has a set of weighted vertices $V$ and a set of undirected edges $E$. 
Each vertex $v \in V$ represents a sharing candidate $(p,Q_p)$ where $p \in S$ is a pattern and $Q_p \subseteq Q$ is the set of all queries containing $p$.
Each vertex is assigned a weight $\mathit{BValue}(p,Q_p)>0$ that corresponds to the benefit value of $(p,Q_p)$ (Equation~\ref{eq:bvalue}).
Each edge $(v,u) \in E$ represents a sharing conflict between the candidates $v, u \in V$. 
\label{def:graph}
\end{definition}

\begin{example}
Figure~\ref{fig:graph} shows the \app\ graph for the traffic monitoring workload in Figure~\ref{fig:queries} and Table~\ref{tab:candidates}. 
\end{example}

\begin{algorithm}[t]
\begin{algorithmic}[1]
\Require A hash table $H$ mapping each sharable pattern $p$ to a list of queries $Q_p$ that contain $p$
\Ensure \app\ graph $G=(V,E)$
\State $V \leftarrow \emptyset;\ E \leftarrow \emptyset;\ G \leftarrow (V,E)$
\ForAll {$p$ in $H$} $Q_p \leftarrow H.\mathit{get}(p)$
    \If {$\mathit{BValue}(p,Q_p)>0$ \textbf{and} $Q_p.\mathit{size} > 1$} 
    	\State $v \leftarrow (p,Q_p);\ v.\mathit{weight} \leftarrow \mathit{BValue}(p,Q_p))$
        \State $V \leftarrow V \cup v$        
   		 \ForAll {$u$ in $V$}
       		\If {$v$ and $u$ are in sharing conflict}
    		\State $E \leftarrow E \cup (v,u)$
			\EndIf
		\EndFor
    \EndIf
\EndFor
\State \Return $G$
\end{algorithmic}
\caption{\app\ graph construction algorithm}
\label{algo:construction}
\end{algorithm}

\textbf{\app\ Graph Construction Algorithm} consumes a hash table $H$ that maps each sharable pattern $p$ to the list of queries $Q$ that contain $p$. 
If a pattern $p$ is beneficial to be shared by at least two queries, the vertex $v=(p, Q_p)$ with weight $\mathit{BValue}(p,Q_p)$ is inserted into the graph (Lines~3--5 in Algorithm~\ref{algo:construction}). Non-beneficial candidates are omitted.
The edges representing the sharing conflicts between $v$ and other vertices in the graph are inserted (Lines~6--8). The graph is returned at the end (Line~9).

\textbf{Data Structures}.
We implement the \app\ graph as an adjacency list to efficiently retrieve the neighbors of a vertex $v$, i.e., identify the sharing conflicts of $v$.
Each vertex stores 
a sharable pattern $p$, 
a list of queries $Q_p$ that contain $p$, and
the benefit value of the sharing candidate $(p,Q_p)$.
The position of a query $q$ in the list $Q_p$ corresponds to the identifier of $q$. Thus, we can conclude whether two candidates are in conflict in linear time in the maximal number of queries $|Q_p|$ containing a sharable pattern $p$, i.e., $O(|Q_p|)$.

\textbf{Complexity Analysis}.
The time complexity is determined by 
the number of sharable patterns $|H|$,
the number of sharing candidates $|V|$, and 
the maximal number of queries $|Q_p|$ containing a pattern.
In total, $O(|H| |V| |Q_p|)$.
The space complexity corresponds to the size of the graph, i.e., $\Theta(|V| |Q_p| + |E|)$.

\section{Sharing Candidate Pruning}
\label{sec:pruning}

Since the search space for an optimal plan is exponential in the number of  candidates (Section~\ref{sec:finder}), we prune two classes of candidates from a \app\ graph.
One, \textit{conflict-ridden candidates} are guaranteed not to be in the optimal plan because their benefit values are too low to counterbalance the loss of benefit from the sharing opportunities they exclude. 
Two, \textit{conflict-free candidates} are guaranteed to 
be in the optimal plan since they do not prevent any other sharing opportunities.

\textbf{Conflict-Ridden Candidates}.
We now map our MESA problem to the problem of finding a Maximum Weight Independent Set (MWIS) in a graph, which is known to be NP-hard~\cite{karp-reducibility}. The greedy algorithm GWMIN~\cite{Sakai2003} for MWIS does not always return a high-quality plan as confirmed by our experiments in Section~\ref{exp:optimizer}. However, its guaranteed minimal weight can be used to prune conflict-ridden candidates.

\begin{definition}(\textbf{Maximum Weight Independent Set})
Let $G=(V,E)$ be a graph with a set of weighted vertices $V$ and a set of edges $E$. For a set of vertices $V'\subseteq V$, we denote the sum of the weights of the vertices in $V'$ as $\mathit{Weight}(V')$.
$\mathit{IS} \subseteq V$ is an {\it independent set} of $G$ if for any vertices $v_i,v_j \in \mathit{IS}$, $(v_i,v_j) \notin E$ holds.
Let $S_{\mathit{IS}}$ be the set of all independent sets of $G$. $\mathit{IS} \in S_{\mathit{IS}}$ is a {\it maximum weight independent set} of $G$ if $\nexists \mathit{IS'} \in S_{\mathit{IS}}$ with $\mathit{Weight}(\mathit{IS'}) > \mathit{Weight}(\mathit{IS})$.
\label{def:MWIS}
\end{definition}

\begin{lemma}
Let $Q$ be a query workload, $G$ be the \app\ graph for $Q$, and $\mathcal{P}_{\mathit{opt}}$ be an optimal sharing plan for $Q$.
Then, $\mathcal{P}_{\mathit{opt}}$ is an MWIS of $G$.
\label{lemma:mapping}
\end{lemma}

\begin{proof} 
By Definitions~\ref{def:valid} and \ref{def:opt-sharing-plan}, $\mathcal{P}_{\mathit{opt}}$ is valid, i.e., contains no conflicting sharing candidates. By Definition~\ref{def:graph}, no vertices in $\mathcal{P}_{\mathit{opt}}$ are connected by an edge in $G$. By Definition~\ref{def:MWIS}, $\mathcal{P}_{\mathit{opt}}$ is an independent set of $G$.
By Definition~\ref{def:opt-sharing-plan}, $\mathcal{P}_{\mathit{opt}}$ has the maximum score among all valid plans. By Definition~\ref{def:graph}, $\mathcal{P}_{\mathit{opt}}$ has the maximum weight among all independent sets of $G$. By Definition~\ref{def:MWIS}, $\mathcal{P}_{\mathit{opt}}$ is an MWIS of $G$.
\end{proof} 

GWMIN returns an independent set $\mathit{IS}$ with weight:
\begin{equation}
\mathit{Weight}(\mathit{IS}) \geq \sum _{u\in V} \frac{\mathit{weight}(u)}{\mathit{degree}(u)+1}
\label{eq:guaranteed_weight}
\end{equation} 

To safely prune a conflict-ridden candidate $v$, we define the maximal score of a plan containing $v$, denoted $\mathit{Score}_{\mathit{max}}(v)$. In best case, a plan containing $v$ includes all other candidates that are not in conflict with $v$. Thus, $\mathit{Score}_{\mathit{max}}(v)$ corresponds to the summation of benefit values of all sharing candidates that are not in conflict with $v$. 

\begin{definition}(\textbf{Maximal Score of a Plan Containing a Sharing Candidate})
Let $v \in V$ be a sharing candidate in a \app\ graph $G=(V,E)$ and $\mathcal{N}(v) \subseteq V$ be the set of candidates that are in conflict with $v$. 
The \textit{maximal score of a sharing plan containing} $v$ is defined as follows:
\begin{equation}
\mathit{Score}_{\mathit{max}}(v) = \sum _{u \in V \setminus \mathcal{N}(v)} \mathit{BValue}(u)
\end{equation}
\label{def:maxscore}
\end{definition}
\vspace*{-2mm}

\begin{lemma}
For a valid sharing plan $\mathcal{P}$ and a sharing candidate $v \in \mathcal{P}$, $\mathit{Score}(\mathcal{P}) \leq \mathit{Score}_{\mathit{max}}(v)$ holds.
\label{lemma:maxscore}
\end{lemma}

\begin{proof}
Let $G=(V,E)$ be the \app\ graph such that $\mathcal{P} \subseteq V$ is an independent set of $G$ and $v \in \mathcal{P}$. By Definition~\ref{def:valid}, $\mathcal{P}$ contains no conflicting candidates. 
By Definition~\ref{def:graph}, all vertices in $\mathcal{N}(v)$ are in conflict with $v$ and thus are not in $\mathcal{P}$. 
Since $\mathcal{P}$ may need to remove additional vertices to avoid other conflicts, $\mathcal{P} \subseteq V \setminus \mathcal{N}(v)$. 
By Definition~\ref{def:maxscore}, $\mathit{Score}_{\mathit{max}}(v)$ is the sum of BValues of all candidates in $V \setminus \mathcal{N}(v)$. 
Since all BValues of vertices in $V$ are positive (Section~\ref{sec:bvalue}),  $\mathcal{P} \subseteq V \setminus \mathcal{N}(v)$ implies $\mathit{Score}(\mathcal{P}) \leq \mathit{Score}_{\mathit{max}}(v)$.
\end{proof}

\begin{definition}(\textbf{Conflict-Ridden Sharing Candidate})
Let $G=(V,E)$ be a \app\ graph. 
A sharing candidate $v \in V$ is \textit{conflict-ridden} if the maximal score of a sharing plan containing $v$ is lower than the guaranteed weight of GWMIN.
\begin{equation}
\mathit{Score}_{\mathit{max}}(v)<\sum _{u \in V} \frac{\mathit{BValue}(u)}{\mathit{degree}(u)+1}
\end{equation}
\label{def:conflict-ridden}
\end{definition}
\vspace*{-2mm}

\textbf{\textit{Conflict-Ridden Candidate Pruning}}. 
All conflict-ridden candidates are pruned from the \app\ graph without sacrificing the optimality of the resulting sharing plan.

\begin{example}
The guaranteed weight on the graph in Figure~\ref{fig:graph} is 
$\frac{25}{6} + \frac{9}{4} + \frac{12}{5} + \frac{15}{4} + \frac{20}{5} + \frac{8}{2} + \frac{18}{1} \approx 38.57$. 
Since $\mathit{Score}_{\mathit{max}}(p_3,\{q_3,q_4\}) = 
\mathit{BValue}(p_3,\{q_3,q_4\}) + 
\mathit{BValue}(p_6,\{q_1,q_5\}) + 
\mathit{BValue}(p_7,\{q_6,q_7\}) = 38 < 38.57$, an optimal sharing plan cannot contain $(p_3,\{q_3,q_4\})$. Thus, this candidate and its conflicts can be pruned.
\label{ex:conflict-ridden}
\end{example}

\textbf{Conflict-Free Candidates} do not exclude any other sharing opportunities and increment the score of a plan by their benefit values. Such candidates can be directly added to an optimal plan and removed from further analysis. 

\begin{definition}(\textbf{Conflict-Free Sharing Candidate})
A sharing candidate $v \in V$ in a \app\ graph $G=(V,E)$ is \textit{conflict-free} if $\not\exists u \in V$ with $(v,u) \in E$.
\label{def:conflict-free}
\end{definition}

\begin{example}
The conflict-free candidate $(p_7,\{q_6,q_7\})$ in Figure~\ref{fig:graph} increments the score of a plan by its benefit 18.
\label{ex:conflict-free}
\end{example}

\begin{figure*}[t]
\centering
\includegraphics[width=0.9\textwidth]{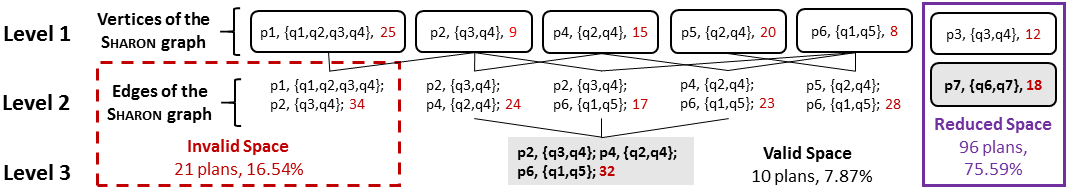}
\caption{Search space of the sharing plan finder algorithm}
\label{fig:ssp}
\end{figure*}

\textbf{\app\ Graph Reduction Algorithm} (Algorithm~\ref{algo:reduction}) consumes a \app\ graph $G$ and the guaranteed weight of GWMIN. It removes all conflict-ridden or conflict-free candidates from the graph $G$. The algorithm returns the reduced graph and the set of conflict-free candidates.

\begin{algorithm}[t]
\begin{algorithmic}[1]
\Require \app\ graph $G$, guaranteed weight $\mathit{min}$ of GWMIN
\Ensure Reduced graph $G$, conflict-free candidates $F$
\State $F \leftarrow \emptyset$
\While {$G$ can be reduced} 
\ForAll {$v$ in $V$}
	\If {degree($v$) = 0}
    	\State $F \leftarrow F \cup v$; $G$.remove($v$)
    \ElsIf {$\mathit{Score}_{\mathit{max}}(v) < \mathit{min}$}
    	\State $G$.remove($v$)
    \EndIf
\EndFor
\EndWhile
\State \Return $G, F$
\end{algorithmic}
\caption{\app\ graph reduction algorithm}
\label{algo:reduction}
\end{algorithm}

\textbf{Complexity Analysis}.
The time complexity is determined by the nested loops that iterate $O(|V|)$ and $\Theta(|V|)$ times respectively. The time complexity of removing a candidate $v$ from the graph in Line~7 is $O(|E|)$ since all conflicts of $v$ are also deleted. Thus, the time complexity is quadratic $O(|V|^2|E|)$.
The space complexity is determined by the size of the graph $G$ and the set $F$. Since $|F| \leq |V|$, the space costs are linear, i.e., $O(|V|+|E|)$.

\begin{example}
Figure~\ref{fig:ssp} depicts the search space for an optimal plan for our running example. Since the conflict-ridden candidate $(p_3,\{q_3,q_4\})$ is pruned (Example~\ref{ex:conflict-ridden}), while the conflict-free candidate $(p_7,\{q_6,q_7\})$ is added to the optimal plan (Example~\ref{ex:conflict-free}), the search space is reduced by $2^{7} - 2^{5} = 96$ plans. This \textbf{\textit{reduced}} space is indicated by a solid frame in Figure~\ref{fig:ssp}. It corresponds to 75.59\% of the search space.
\label{ex:reduced}
\end{example}

\section{Sharing Plan Finder}
\label{sec:finder}

Based on the reduced \app\ graph, we now propose the optimal sharing plan finder. In addition to the non-beneficial and conflict-ridden candidate pruning principles, we define invalid branch pruning. It cuts off those branches of the search space that contain only invalid plans early on.

\textbf{Search Space for an Optimal Sharing Plan}. 
The parent-child relationships between sharing plans are defined next.

\begin{definition}(\textbf{Parent-Child Relationship between Sharing Plans})
Let $\mathcal{P}$ and $\mathcal{P}'$ be sharing plans. If $\mathcal{P} \subset \mathcal{P}'$, then $\mathcal{P}$ is an \textit{ancestor} of $\mathcal{P}'$ ($\mathcal{P}'$ is a \textit{descendant} of $\mathcal{P}$). If $|\mathcal{P}| = |\mathcal{P}'| - 1$, then $\mathcal{P}$ is a \textit{parent} of $\mathcal{P}'$ ($\mathcal{P}'$ is a \textit{child} of $\mathcal{P}$).
\label{def:parent-child}
\end{definition}

The search space has a lattice shape (Figure~\ref{fig:ssp}). The plans in Level~1 correspond to the vertices in Figure~\ref{fig:graph}. Level~$s$ contains sharing plans of size $s$. 
The \textit{size of the search space} is exponential in the number of candidates, denoted $|V|$. It is computed as the sum of the number of plans at each level:  
\begin{equation}
\sum _{s=0}^{|V|} \binom{|V|}{s} = 2^{|V|}
\label{eq:complexity}
\end{equation}

\begin{lemma}
$Score(\mathcal{P}')>Score(\mathcal{P})$ if $\mathcal{P}$ is a parent of $\mathcal{P}'$.
\end{lemma}

\begin{proof}
By Definition~\ref{def:parent-child}, $\mathcal{P} \subset \mathcal{P}'$ and $|\mathcal{P}| = |\mathcal{P}'| - 1$. 
Let $\mathcal{P}' \setminus \mathcal{P} = (p,Q_p)$.
By Definition~\ref{def:graph}, only a candidate $(p,Q_p)$ with $\mathit{BValue}(p,Q_p)>0$ is included into a \app\ graph.
Thus, $(p,Q_p)$ increases the score of $\mathcal{P}'$ compared to $\mathcal{P}$. 
\end{proof}

A naive plan finder considers all combinations of candidates and keeps track of a valid plan with the maximal score seen so far. However, this solution constructs many invalid plans that are subsequently discarded. 
To avoid such exhaustive search, we prove the following properties of the search space. 

\begin{lemma}
All descendants of an invalid plan are invalid.
\label{lemma:invalid}
\end{lemma}

\begin{proof}
Let $\mathcal{P}$ be an invalid sharing plan and $\mathcal{P}_d$ be its descendant. By Definition~\ref{def:parent-child}, $\mathcal{P} \subset \mathcal{P}_d$. Thus, $\mathcal{P}_d$ ``inherits" all sharing conflicts from $\mathcal{P}$ which makes $\mathcal{P}_d$ invalid.
\end{proof}

\textbf{\textit{Invalid Branch Pruning}}.
Invalid plans of size two correspond to edges of a \app\ graph (Figures~\ref{fig:graph} and \ref{fig:ssp}). Thus, all descendants of invalid plans of size two can be safely pruned. Our plan finder cuts off invalid branches at their roots.

\begin{example}
In Figure~\ref{fig:ssp}, only 7.87\% of the search space is \textbf{\textit{valid}}. It consists of 10 plans. This valid space is traversed to find the optimal plan $\{(p_2,\{q_3,q_4\});\ (p_4,\{q_2,q_4\});\ (p_6,$ $\{q_1,q_5\});\ (p_7,\{q_6,q_7\})\}$ highlighted by a darker background.

In Figure~\ref{fig:ssp}, 16.54\% of the search space is \textbf{\textit{invalid}}. The invalid space consists of 21 plans = $2^{5}$ not reduced plans -- 10 valid plans -- 1 empty plan. The invalid space is indicated by the dashed frame. It is pruned by our plan finder.
%
%
\end{example}

\textbf{Valid Search Space Traversal}.
A plan of size one is valid by Definition~\ref{def:valid}. A plan of size two $\{v_1,v_2\}$ is valid if there is no edge $(v_1,v_2)$ in the \app\ graph. Validity of a larger plan is determined as described next. 

\begin{lemma}
A sharing plan $\mathcal{P}$, $|\mathcal{P}|>2$, is valid if and only if all its parents are valid.
\label{lemma:valid}
\end{lemma}

\begin{proof}
``$\Rightarrow$" The proof follows directly from Lemma~\ref{lemma:invalid}.

``$\Leftarrow$" Let $\mathcal{P}$ contain a sharing conflict, say, between candidates $v$ and $u$. Then there exists a parent $\mathcal{P}_p$ of $\mathcal{P}$ that contains $v$ and $u$. Hence, $\mathcal{P}_p$ is invalid.
\end{proof}

By Definition~\ref{def:parent-child}, a plan of size $s$ has $s$ parents. Instead of accessing all parent plans to generate one new valid plan, we prove that only two parents and one ancestor of size two must be valid to guarantee validity of a sharing plan (similarly to Apriori candidate generation~\cite{apriori94VLDB}). 

\begin{lemma}
Let $G=(V,E)$ be a \app\ graph,
$\mathcal{P}_1$ and $\mathcal{P}_2$ be valid parents of $\mathcal{P},\, |\mathcal{P}|>2$.
For two candidates $v_1 = \mathcal{P}_1 \setminus \mathcal{P}_2$ and $v_2 = \mathcal{P}_2 \setminus \mathcal{P}_1$,
if $(v_1,v_2) \not\in E$, then $\mathcal{P}$ is valid.
\label{lemma:three}
\end{lemma}

\begin{proof}
Assume all the above conditions hold but $\mathcal{P}$ is invalid. Then $\mathcal{P}$ contains at least one pair of conflicting candidates. By Definition~\ref{def:parent-child}, $\mathcal{P} = \mathcal{P}_1 \cup \mathcal{P}_2$ and $\mathcal{P}$ has one additional candidate compared to $\mathcal{P}_1$ (or $\mathcal{P}_2$). Since $\mathcal{P}_1$ and $\mathcal{P}_2$ are valid, there can be only one pair of conflicting candidates $v_1$ and $v_2$ in $\mathcal{P}$ such that $v_1 = \mathcal{P}_1 \setminus \mathcal{P}_2$ and $v_2 = \mathcal{P}_2 \setminus \mathcal{P}_1$. By Definition~\ref{def:graph}, $(v_1,v_2) \in E$ which is a contradiction.
\end{proof}

\begin{algorithm}[t]
\begin{algorithmic}[1]
\Require \app\ graph $G=(V,E)$, set of sharing plans of size $s$ $\mathit{Parents}=\{P_0,\ldots,P_{s-1}\}$
\Ensure Set of sharing plans of size $s+1$ $\mathit{Children}$
\State $\mathit{getNextLevel}(G,\mathit{Parents})$ \{
\State $\mathit{Children} \leftarrow \emptyset$
\ForAll {$(i=0; i<s; i\text{++})$}
 \ForAll {$(j=i+1; j<s; j\text{++})$}
  \If {$s=1$ \textbf{and} $(P_i.v_1, P_j.v_1)$ not in $E$}
   \State $\mathit{Children}$.add($P_i \cup P_j$)
  \EndIf 
  \If {$P_i.v_1=P_j.v_1,\ldots,P_i.v_{s-1}=P_j.v_{s-1}$ \textbf{and} $(P_i.v_s, P_j.v_s)$ not in $E$}
     \State $\mathit{Children}$.add($P_i \cup P_j$)
  \EndIf
 \EndFor
\EndFor
\State \Return $\mathit{Children}$ \}
\end{algorithmic}
\caption{Level generation algorithm}
\label{node-gen}
\end{algorithm}

\textbf{Level Generation Algorithm} consumes a \app\ graph $G$ and a set of sharing plans of size $s$, called $\mathit{Parents}$. It returns level $s+1$ of the search space, i.e., the set of all sharing plans of size $s+1$, called $\mathit{Children}$. 
Algorithm~\ref{node-gen} iterates through all pairs of parent plans of size $s$ (Lines~3--4). 
In the base case, the $\mathit{Parents}$ are the vertices of $G$ and the $\mathit{Children}$ are non-adjacent pairs of vertices (Lines~5--6).
In the inductive case, to generate a valid plan of size $s+1$, the algorithm identifies two plans of size $s$, $P_i$ and $P_j$, that begin with the same $s-1$ decisions. 
If the plan containing the last decisions of $P_i$ and $P_j$ (denoted $P_i.v_s$ and $P_j.v_s$) is valid, the plan $P_i \cup P_j$ is also valid (Lines~7--8).

\textbf{Complexity Analysis}.
The time complexity of Algorithm~\ref{node-gen} is determined by the number of plans at one level, namely, the binomial coefficient $\binom{|V|}{s}$ in Equation~\ref{eq:complexity}. Due to two nested loops, the time complexity is $O(\binom{|V|}{s}^2)$. The space complexity is also determined by the number of plans at one level, i.e., $O\binom{|V|}{s}$.

\begin{figure}[t]
\centering
\begin{minipage}{0.15\columnwidth}
	\centering
    \includegraphics[width=1.0\columnwidth]{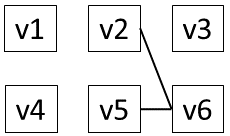}
    \caption{\app\ graph}
    \label{fig:node-gen-graph}
\end{minipage}
\hspace*{2mm}
\begin{minipage}{0.5\columnwidth}
	\centering
	\includegraphics[width=1.0\columnwidth]{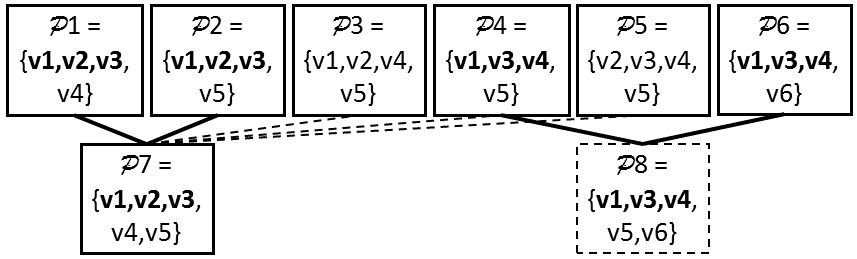}
    \caption{Generation of a new valid sharing plan}
    \label{fig:node-gen-ssp}
\end{minipage}
\end{figure}

\begin{example}
Figure~\ref{fig:node-gen-ssp} shows a portion of a search space with valid plans $\mathcal{P}_1$--$\mathcal{P}_6$ of size four. $\mathcal{P}_7$ is the only valid plan of size five. It is generated as follows. 
(1)~We identify two plans of size four that start with the same three candidates, e.g., $\mathcal{P}_1$ and $\mathcal{P}_2$ start with $\{v_1,v_2,v_3\}$. 
(2)~We compute their symmetric difference $\mathcal{P}_1 \Delta \mathcal{P}_2 = \{v_4,v_5\}$. 
(3)~Since there is no edge $(v_4, v_5)$ in Figure~\ref{fig:node-gen-graph}, $\mathcal{P}_7$ is valid. 
There is no need to check the other three parents of $\mathcal{P}_7$.
In contrast, $\mathcal{P}_8$ is invalid since $v_5$ and $v_6$ are in conflict. 
\label{ex:generation}
\end{example}

\begin{algorithm}[t]
\begin{algorithmic}[1]
\Require \app\ graph $G=(V,E)$, set of conflict-free candidates $F$
\Ensure Optimal sharing plan $\mathit{opt} \cup F$
\State $\mathit{opt} \leftarrow \emptyset;\ \mathit{max} \leftarrow 0$
\ForAll {$v$ in $V$}
	\If {$\mathit{BValue}(v) > \mathit{max}$}
 	\State $\mathit{opt} \leftarrow \{v\};\ \mathit{max} \leftarrow \mathit{BValue}(v)$
 	\EndIf
\EndFor
\State $\mathit{Level} \leftarrow \mathit{getNextLevel}(G,V)$
\While {$\mathit{Level}\neq\emptyset$}
 	\ForAll {$P$ in $\mathit{Level}$}
  		\If {$\mathit{Score}(P) > \mathit{max}$}
   		\State $\mathit{opt} \leftarrow P;\ \mathit{max} \leftarrow \mathit{Score}(P)$
  		\EndIf
 	\EndFor
\State $\mathit{Level} \leftarrow \mathit{getNextLevel}(G,\mathit{Level})$
\EndWhile
\State \Return $\mathit{opt} \cup F$
\end{algorithmic}
\caption{Sharing plan finder algorithm}
\label{bb}
\end{algorithm}

\textbf{Sharing Plan Finder Algorithm} traverses valid search space using Breadth-First-Search. Algorithm~\ref{bb} effectively prunes invalid branches at their roots. It constructs only valid plans and returns an optimal plan among them. 

\textbf{Correctness}.
We prove that Algorithm~\ref{bb} considers all valid sharing plans, i.e., it returns the optimal sharing plan.

\begin{lemma}
If a sharing plan is valid, then it is considered by the sharing plan finder algorithm.
\label{lemma:correctness}
\end{lemma}

\begin{proof}
We prove Lemma~\ref{lemma:correctness} by induction. The base cases are $s=1$ and $s=2$. First, $V$ is the set of all valid sharing plans of size 1. It is considered by Algorithm~\ref{bb}. Second, Algorithm~\ref{bb} (Line~5) generates all plans of size 2 by considering all non-adjacent vertex pairs in Algorithm~\ref{node-gen} (Lines~5--6).

Now, we assume that all valid plans of up to and including size $s$, such that $s\geq 2$, are considered. 
We will prove that all valid plans of size $s+1$ are also considered. 
Let $\mathcal{P}=\{ v_1,\ldots,v_s,v_{s+1}\}$ be a valid plan of size $s+1$. 
Then $\mathcal{P}_1=\{ v_1,\ldots,v_{s-1},v_s\}$, 
$\mathcal{P}_2=\{ v_1,\ldots,v_{s-1},v_{s+1}\}$, and 
$\mathcal{P}_3=\{v_s, v_{s+1}\}$ are ancestors of $\mathcal{P}$. By Lemma~\ref{lemma:valid}, $\mathcal{P}_1,\mathcal{P}_2,$ and $\mathcal{P}_3$ are valid. By the induction assumption, they were considered by Algorithm~\ref{bb}. Algorithm~\ref{node-gen} generates the plan $\mathcal{P}$ from $\mathcal{P}_1,\mathcal{P}_2,$ and $\mathcal{P}_3$. Thus, Algorithm~\ref{bb} considers $\mathcal{P}$.
\end{proof}

\textbf{Data Structures}.
For each plan $\mathcal{P}$, we store the list of sharing candidates and the score of $\mathcal{P}$. These candidates are sorted alphabetically by their patterns within a plan because sequential access of candidates in plans enables efficient generation of new valid plans (Lemma~\ref{lemma:three}, Example~\ref{ex:generation}).
Our plan finder stores the best plan found so far and only one level of the search space at a time. It discards a level immediately after generating the next level.

\textbf{Complexity Analysis}.
Since the entire valid search space is traversed, the algorithm has exponential time and space complexity in the worst case (Equation~\ref{eq:complexity}).
However, the \app\ optimizer is efficient on average thanks to its effective pruning principles (Section~\ref{exp:optimizer}).

\textbf{Effectiveness of Sharing Plan Finder}.
Our experiments in Section~\ref{exp:optimizer} demonstrate 
that our \app\ optimizer finds an optimal plan in reasonable time due to three effective pruning principles (i.e., non-beneficial, conflict-ridden, and invalid candidates in Sections~\ref{sec:bvalue}, \ref{sec:pruning}, and \ref{sec:finder}). Only in the following two extreme cases our solution may be ineffective:

1) Since the algorithm traverses the entire valid space, its time complexity is exponential (Equation~\ref{eq:complexity}). 
If the search space is too large, we can constrain the optimization time by $l$ seconds. If our \app\ optimizer does not return an optimal plan within $l$ seconds, we instead would run GWMIN~\cite{Sakai2003} with polynomial time complexity to find a sharing plan and start our \app\ executor using this plan. Later, when the optimal plan is produced by our \app\ optimizer, we can replace the greedily found plan by the optimal plan. 

2) The valid search space becomes small if many sharing conflicts exist (Figure~\ref{fig:ssp}). In this case, only a few patterns can be shared and a fairly low score of a sharing plan would be achieved. In the worst case, no pattern can be shared, i.e., \app\ defaults to the Non-Shared Method (Section~\ref{sec:non-shared}). Our optimizer finds such a trivial plan very quickly.

\textbf{Optimal versus Greedily Chosen Plan}.
While the greedy algorithm GWMIN is useful to reduce the search space (Section~\ref{sec:pruning}), the score of a greedily chosen plan might be considerably lower than the score of an optimal plan. 

\begin{example}
Even in our small example in Figure~\ref{fig:graph}, the greedily chosen plan
$\mathcal{P}_{\mathit{gre}}$ = 
$\{(p_1,\{q_1,q_2,q_3,q_4\})$;
$(p_7,\{q_6,q_7\})\}$ 
has score 43, while the optimal plan
$\mathcal{P}_{\mathit{opt}}$ = 
$\{(p_2,\{q_3,q_4\})$; 
$(p_4,\{q_2,q_4\}$; 
$(p_6, \{q_1,q_5\})$;
$(p_7,\{q_6,q_7\})\}$ 
increases $Score(\mathcal{P}_{gre})$ by more than 16\% to 50. 
\label{ex:optimal-vs-greedy}
\end{example}

\section{Extensions of the Sharon Approach}
\label{sec:discussion}

In this section, we briefly describe the extensions of our approach to relax the simplifying assumptions in Section~\ref{sec:model}.

\subsection{Sharing Conflict Resolution}

Our analysis in Section~\ref{sec:graph} reveals that promising sharing opportunities might be excluded by sharing conflicts. Generally, the more queries share a pattern the higher the probability of sharing conflicts becomes (Definition~\ref{def:sharing-conflict}). We now open up additional sharing opportunities by resolving sharing conflicts as follows.

Given a \app\ graph $G=(V,E)$, we expand each candidate $v=(p,Q_p) \in V$ with conflicts $E_v \subseteq E$ to a set of options $\mathcal{O}_p$. Each option $v' = (p,Q'_p) \in \mathcal{O}_p$ resolves a different subset of conflicts $E_v' \subseteq E_v$ of the original candidate $v$ with other candidates $u \in V \setminus O_p$. In contrast to the original candidate $v$, an option $v'$ considers sharing the pattern $p$ by a \textit{subset} of queries containing $p$, i.e., $Q'_p \subseteq Q_p,$ $|Q'_p| > 1$.

\begin{example}
In Figure~\ref{fig:graph}, the sharing candidate $(p_1,\{q_1,$ $q_2,q_3,q_4\})$ can be expanded to a set of options. The option $(p_1,\{q_1,q_3\})$ is not in sharing conflict with the candidates $(p_4,\{q_2,q_4\})$ and $(p_5,\{q_2,q_4\})$. Thus, they could belong to the same sharing plan which may have a higher score than a plan containing the original candidate $(p_1,\{q_1,q_2,q_3,q_4\})$.
\end{example}

\begin{definition}(\textbf{Resolved Sharing Conflict}.)
Let the candidates $v_1 = (p_1,Q_1)$ and $v_2=(p_2,Q_2) \in V$ be in conflict $(v_1,v_2) \in E$ caused by the queries $Q = Q_1 \cap Q_2$ such that $Q = Q'_1 \cupdot Q'_2$.%
\footnote{$\cupdot$ denotes disjoint set union, meaning that $Q = Q'_1 \cup Q'_2$ but $Q'_1 \cap Q'_2 = \emptyset$.}
The conflict $(v_1,v_2)$ is \textit{resolved} by omitting $Q'_1$ and $Q'_2$ from $Q_1$ and $Q_2$ respectively.
\label{def:resolved-conflict}
\end{definition}

By Definition~\ref{def:sharing-conflict}, the sharing candidates $v'_1=(p_1,Q_1 \setminus Q'_1)$ and $v'_2=(p_2,Q_2 \setminus Q'_2)$ are \textit{not in conflict} since $(Q_1 \setminus Q'_1) \cap (Q_2 \setminus Q'_2) = \emptyset$.
The conflict $(v_1,v_2)$ is resolved if \textit{any} query sets $Q'_1$ and $Q'_2$ that compose $Q$ are omitted from $Q_1$ and $Q_2$ respectively. In the worst case, \textit{all} combinations of queries $Q$ are included into the sets of options for $v_1$ and $v_2$.

\begin{algorithm}[t]
\begin{algorithmic}[1]
\Require \app\ graph $G=(V,E),\ v = (p,Q_p) \in V$
\Ensure Set $O_p$ of sharing candidate options for $p$
\State $\mathit{getSet}(G,v)$ \{
\State $L_c,L_n \leftarrow$ empty stacks; $L_c.\mathit{push}(v);\ O_p \leftarrow \{v\}$
\While {$!L_c.\mathit{isEmpty}()$}
\State $v \leftarrow L_c.\mathit{pop}()$
\ForAll {conflict $(v,u)$ in $E$}
	\State $Q_c \leftarrow$ queries in $Q_p$ that cause $(v,u)$
    \ForAll {combination $C$ of $Q_c$ that can resolve $(v,u)$}
    	\State $Q'_p \leftarrow Q_p \setminus C$
		\If {$|Q'_p| > 1$ \textbf{and} $Q'_p$ is new}
    		\State $v' \leftarrow (p,Q'_p);\ L_n.\mathit{push}(v');$
            \State $O_p \leftarrow O_p \cup \{ v' \}$
    	\EndIf	
    \EndFor
\EndFor
\If {$L_c.\mathit{isEmpty}()$}
\State $L_c \leftarrow L_n;\ L_n \leftarrow$ empty stack
\EndIf
\EndWhile
\State \Return $O_p\ \}$
\end{algorithmic}
\caption{Sharing candidate expansion algorithm}
\label{algo:set-generation}
\end{algorithm}

\textbf{Sharing Candidate Expansion Algorithm}.
For a \app\ graph $G$ and a candidate $v=(p,Q_p) \in V$, Algorithm~\ref{algo:set-generation} builds a tree of options $O_p$ using Breadth First Search. The root of this tree is the original candidate $v$. To generate a child of $v$, the algorithm skips the queries from $Q_p$ that cause a conflict of $v$ with another sharing candidate $u \in V \setminus O_p$. We label an edge between $v$ and its child by the sharing candidate $u$. The algorithm terminates when no new option with at least two queries can be generated.

\textbf{Complexity Analysis}.
The time and space complexity of Algorithm~\ref{algo:set-generation} are determined by the maximal size of a set $|O_p^{max}|$.
Let $d$ be the maximal degree of a candidate $v \in V$ and $k$ be the maximal number of queries that cause a conflict. 
For each conflict $(v,u) \in E$, all combinations of queries causing this conflict are considered (nested for-loops in Lines~5--10 and 7--10). Thus, 
\begin{equation}
|O_p^{max}| = \sum_{i=0}^{d} \binom{d}{i} \sum_{j=0}^{k-1} \binom{k}{j}
\label{eq:max-set-size}
\end{equation}
where $i$ denotes the number of resolved conflicts, while $j$ corresponds to the number of skipped queries to resolve one conflict.

\begin{figure}[t]
\centering
\includegraphics[width=0.45\columnwidth]{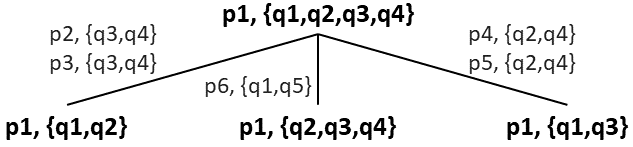}
\caption{Sharing candidate options for pattern $p_1$}
\label{fig:set-generation}
\end{figure}

\begin{example}
Figure~\ref{fig:set-generation} illustrates the sharing candidate options for the candidate $v = (p_1,\{q_1,q_2,q_3,q_4\})$ in Figure~\ref{fig:graph}.
To resolve the conflict with $u_1=(p_2,\{q_3,q_4\})$ and $u_2=(p_3,\{q_3,q_4\})$, queries $q_3$ and $q_4$ are dropped from the set of queries of $v$. 
The edge between $v$ and its child $(p_1,\{q_1,q_2\})$ is labeled by $u_1,u_2$.
Other conflicts of $v$ are resolved analogously.
\end{example}

\begin{algorithm}[t]
\begin{algorithmic}[1]
\Require \app\ graph $G=(V,E)$
\Ensure Expanded \app\ graph $G$
\State $V' \leftarrow \emptyset;\ E' \leftarrow \emptyset$
\ForAll {$v=(p,Q_p)$ in $V$}
	\State $O_p \leftarrow \mathit{getSet}(G,v);\ V' \leftarrow V' \cup O_p$
    \ForAll {$v'$ in $O_p$}
    \ForAll {$u$ in $V'$}
		\If {$v'$ and $u$ are in sharing conflict}
    	\State $E'.\mathit{add}(v',u)$
		\EndIf
	\EndFor
    \EndFor    
\EndFor
\State \Return $G \leftarrow (V',E')$
\end{algorithmic}
\caption{Sharing conflict resolution algorithm}
\label{algo:resolution}
\end{algorithm}

\textbf{Sharing Conflict Resolution Algorithm}. 
For a \app\ graph $G$ and each candidate $v=(p,Q_p) \in V$, (Algorithm~\ref{algo:resolution}) expands $v$ to a set of options $O_p$ to open up additional sharing opportunities. The algorithm updates the conflicts of these options and returns the expended graph.

\textbf{Complexity Analysis}.
The time complexity is determined by three nested for-loops that are called
$\Theta(|V|)$, $|O_p^{\mathit{max}}|$ and $\Theta(|V'|)$ times respectively 
where $|O_p^{\mathit{max}}|$ denotes the maximal size of a set (Equation~\ref{eq:max-set-size}). 
Since $|V| \leq |V'|$ and $|O_p^{\mathit{max}}| \leq |V'|$, the time complexity is cubic in the number of candidates in the expanded \app\ graph in the worst case, i.e., $O(|V'|^3)$.
The space complexity is determined by the size of the expanded graph, i.e., $\Theta(|V'|+|E'|)$.

\begin{figure}[t]
\centering
\includegraphics[width=0.35\columnwidth]{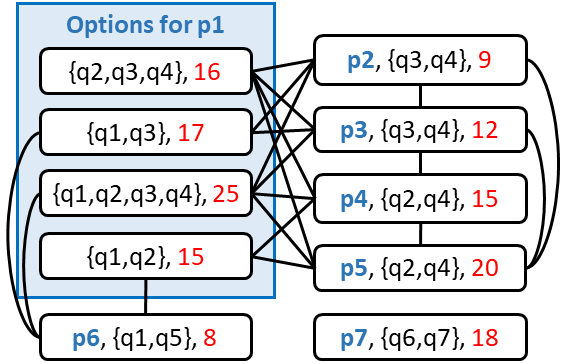}
\caption{Expanded \app\ graph}
\label{fig:expanded}
\end{figure}

\begin{example}
The \app\ graph in Figure~\ref{fig:graph} is expanded in Figure~\ref{fig:expanded}. 
The sharing candidate for pattern $p_1$ is expanded into a set of options and highlighted by a rectangle frame. 
The sets for other candidates contain only the original candidate. Conflicts within sets are omitted for readability. 
\label{ex:expanded-graph}
\end{example}

The expanded graph is then reduced (Section~\ref{sec:pruning}) and serves as input to our sharing plan finder (Section~\ref{sec:finder}).

\subsection{Different Predicates, Grouping, and Windows} 

Leveraging existing techniques, our \app\ approach can share event sequence aggregation among queries with different grouping, windows, and predicates. 
Grouping partitions the stream into sub-streams by the values of grouping attributes~\cite{QCRR14, GSCL12}. Windows and predicates further partition these sub-streams into disjoint segments and share the intermediate aggregates per segment to compute the final results for each query~\cite{GSCL12, KWF06, AW04, LMTPT05}. 
These refinement strategies might not always be effective, because of a large number of small segments and the overhead of their computation. However, these are orthogonal problems. Our \app\ approach can be applied within each segment to tackle different query patterns.

\subsection{Multiple Occurrences of an Event Type in a Pattern}

If an event type $E$ occurs $k$ times in a pattern, an event of type $E$ updates the counts of $k$ prefix patterns that end at $E$ (Section~\ref{sec:benefit}). Then, the time complexity of both the Non-Shared and the Shared methods increases by the multiplicative factor $k$ (Equations~\ref{eq:non-shared-1-query}, \ref{eq:comp}, and \ref{eq:shared-k-queries}). Our \app\ optimizer is not affected by this extension.

\subsection{Dynamic Workloads}

In dynamic environments, new queries may be added or existing queries may be removed. Even if the queries remain the same, the workload may still vary due to event rate fluctuations. Thus, a chosen plan may become sub-optimal. In this case, our \app\ approach leverages runtime statistics techniques~\cite{LeiR14} to detect such fluctuations and to trigger the \app\ optimizer to produce a new optimal plan based on the new workload. Dynamic plan migration techniques~\cite{KWF06, ZhuRH04} can be employed to migrate from the old to the new sharing plan and ensure that no results are lost or corrupted for stateful operators such as aggregation.

\section{Performance Evaluation}
\label{sec:evaluation}

\subsection{Experimental Setup}

\textbf{Infrastructure}.
We have implemented our \app\ approach in Java with JRE 1.7.0\_25 running on Ubuntu 14.04 with 16-core 3.4GHz CPU and 128GB of RAM. 
We execute each experiment three times and report the average  here.

\textbf{Data Sets}. 
We evaluate the performance of our \app\ approach using the following data sets.

$\bullet$~\textbf{\textit{TX: New York City Taxi and Uber Real Data Set}}. 
We use the real data set~\cite{uber1} (330GB) containing 1.3 billion taxi and Uber trips in New York City in 2014--2015. Each event carries pick-up and drop-off locations and time stamps in seconds, number of passengers, price, and payment method.

$\bullet$~\textbf{\textit{LR: Linear Road Benchmark Data Set}}. 
We use the traffic simulator of the Linear Road benchmark~\cite{linear_road} for streaming systems to generate a stream of position reports from cars for 3 hours. Each position report carries a time stamp in seconds, a car identifier, its location and speed. Event rate gradually increases from few dozens to 4k events per second. 

$\bullet$~\textbf{\textit{EC: E-Commerce Synthetic Data Set}}. 
Our stream generator creates sequences of items bought together for 3 hours. Each event carries a time stamp in seconds, item and customer identifiers. We consider 50 items and 20 users. The values of item and customer identifiers of an event are randomly generated. The stream rate is 3k events per second.

We ran each experiment on the above three data sets. Due to space limitations, similar charts are not shown here.

\textbf{Event Queries}. 
We evaluate a workload similar to $q_1$--$q_7$ in Section~\ref{sec:introduction} against the taxi and Linear Road data sets and a workload similar to $q_8$--$q_{11}$ against the e-commerce data set. 
Based on our cost model (Section~\ref{sec:benefit}), we vary the major cost factors, namely, number of queries, the length of their patterns, and the number of events per window. Unless stated otherwise, we evaluate 20 queries. The default length of their patterns is 10. The default number of events per window is 200k.

\textbf{Methodology}. 
We run two sets of experiments.

1) \textbf{\textit{Sharon Executor vs. State-of-the-Art Approaches}} (Section~\ref{exp:SOA}). We demonstrate the effectiveness of our \app\ executor (Section~\ref{sec:benefit}) by comparing it to the state-of-the-art techniques A-Seq~\cite{QCRR14}, SPASS~\cite{RLR16}, and Flink~\cite{flink} covering the spectrum of approaches to event sequence aggregation (Figure~\ref{fig:spectrum}). While Section~\ref{sec:related_work} is devoted to a detailed discussion of these approaches, we briefly sketch their main ideas below.

$\bullet$~\textbf{\textit{A-Seq}}~\cite{QCRR14} avoids sequence construction by incrementally maintaining a count for each prefix of a pattern. However, it has no optimizer to determine which queries should share the aggregation of which patterns. By default, it computes each query independently from other queries and thus suffers from repeated computations (Section~\ref{sec:non-shared}).

$\bullet$~\textbf{\textit{SPASS}}~\cite{RLR16} defines shared event sequence construction. Their aggregation is computed afterwards and is not shared. Thus, SPASS is a two-step and only partially shared approach.

$\bullet$~\textbf{\textit{Flink}}~\cite{flink} is a popular open-source streaming system that supports event pattern matching and aggregation. We express our queries using Flink operators. Flink constructs all event sequences prior their aggregation. It does not share computations among different queries.

To achieve a fair comparison, we have implemented A-Seq and SPASS on top of our platform. We execute Flink on the same hardware as our platform.

2) \textbf{\textit{Sharon Optimizer}} (Section~\ref{exp:optimizer}). We study the efficiency of our \app\ optimizer (Sections~\ref{sec:graph}--\ref{sec:discussion}) by comparing it to the greedy algorithm GWMIN~\cite{Sakai2003} and to exhaustive search. We also compare the quality of a greedily chosen plan returned by GWMIN to an optimal plan returned by our \app\ optimizer and the exhaustive search.

\textbf{Metrics}. 
We measure the following metrics common for streaming systems. 
\textbf{\textit{Latency}} is measured in milliseconds as the average time difference between the time point of aggregate output and the arrival time of the latest event that contributed to this result.
\textbf{\textit{Throughput}} corresponds to the average number of events processed by all queries per second.
\textbf{\textit{Peak memory}} is measured in bytes.
For event sequence aggregation algorithms, it is the maximal memory for storing aggregates, events, and event sequences.
For the optimizer algorithms, the peak memory is the maximal memory for storing the \app\ graph and the sharing plans during space traversal.

\subsection{Sharon Executor versus State-of-the-Art Approaches}
\label{exp:SOA}

\begin{figure}[t]
	\centering
    \subfigure[Latency]{
    	\includegraphics[width=0.23\columnwidth]{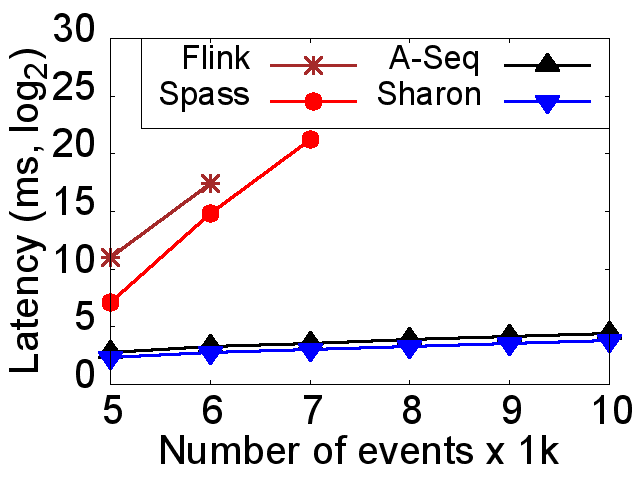}
    	 \label{fig:two-steps-latency}
	}
	\subfigure[Throughput]{
    	\includegraphics[width=0.23\columnwidth]{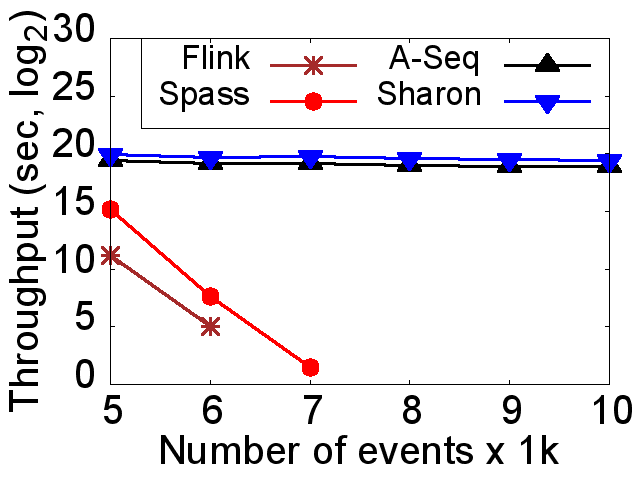}
    	\label{fig:two-steps-thru}
	}
    \vspace*{-2mm}
	\caption{Two-step versus online approaches (Linear Road data set)}
	\label{fig:two-steps}
\end{figure}

\textbf{Two-step Approaches}.
In Figure~\ref{fig:two-steps}, we vary the number of events per window and measure latency and throughput of the event sequence aggregation approaches using the Linear Road benchmark data set. Latency of the two-step approaches (SPASS and Flink) increases exponentially, while throughput decreases exponentially in the number of events. 

\textbf{\textit{SPASS}} achieves 6--fold speed-up compared to Flink for 6k events per window because SPASS shares event sequence construction. Due to event sequence construction overhead, SPASS does not terminate when the number of events exceeds 7k. These measurements are not shown in Figure~\ref{fig:two-steps}.

\textbf{\textit{Flink}} not only constructs all event sequences but also computes each query independently from other queries in the workload. Flink fails for more than 6k events per window. 

The event sequence construction step has polynomial time complexity in the number of events~\cite{ZDI14, QCRR14} and may jeopardize real-time responsiveness for high-rate event streams (Figure~\ref{fig:two-steps}). Thus, these two-step approaches cannot be effective for time-critical processing of high-rate streams.

\begin{figure*}[t]
	\centering
    \subfigure[Latency (TX)]{
    	\includegraphics[width=0.23\columnwidth]{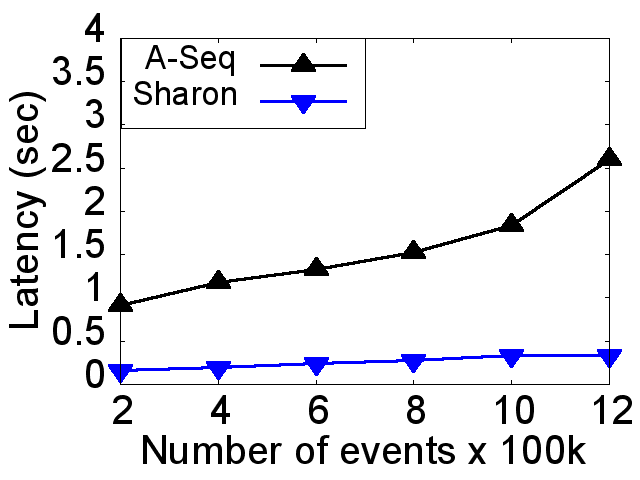}
    	 \label{fig:latency-r}
	}
	\subfigure[Latency (LR)]{
    	\includegraphics[width=0.23\columnwidth]{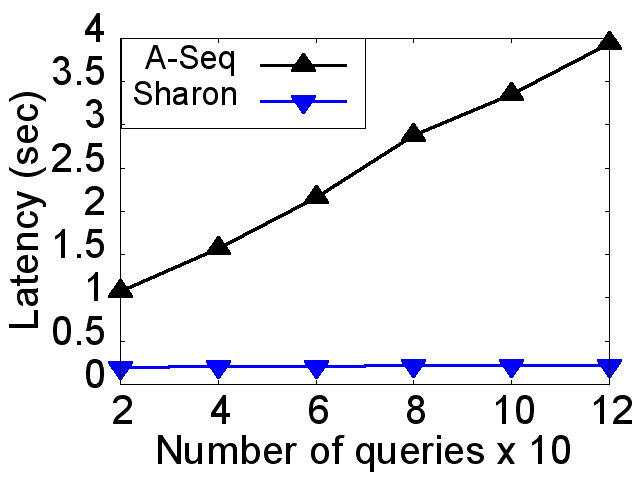}
    	\label{fig:latency-k}
	}
	\subfigure[Latency (EC)]{
    	\includegraphics[width=0.23\columnwidth]{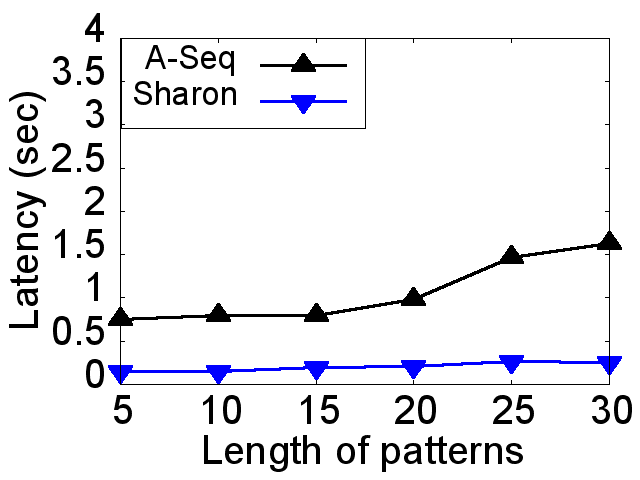}
    	 \label{fig:latency-l}
	}
	\subfigure[Memory (LR)]{
    	\includegraphics[width=0.23\columnwidth]{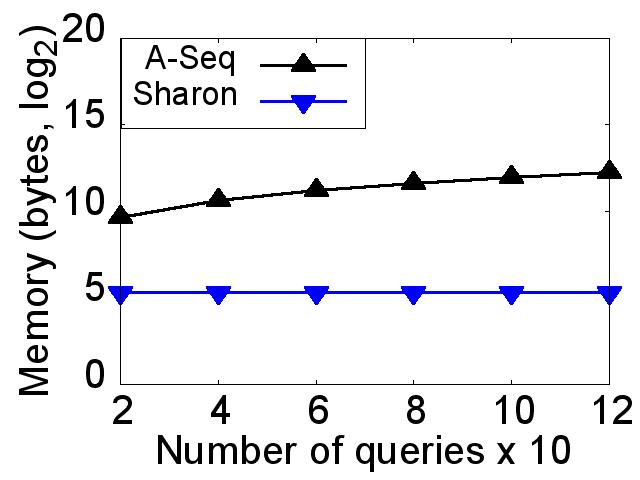}
    	\label{fig:mem-k}
	}
	\subfigure[Throughput (TX)]{
    	\includegraphics[width=0.23\columnwidth]{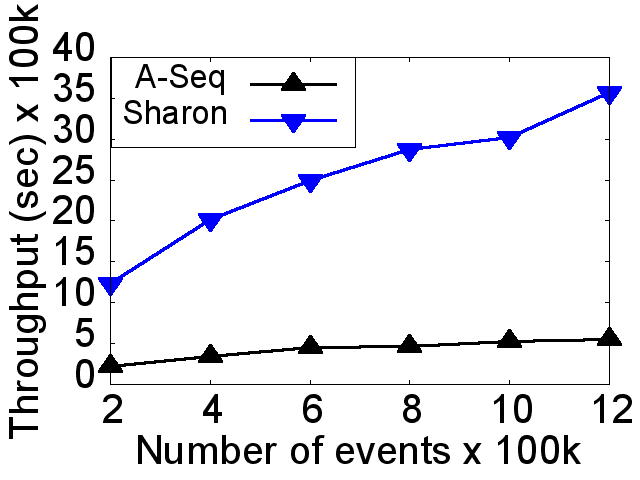}
    	 \label{fig:thru-r}
	}
	\subfigure[Throughput (LR)]{
    	\includegraphics[width=0.23\columnwidth]{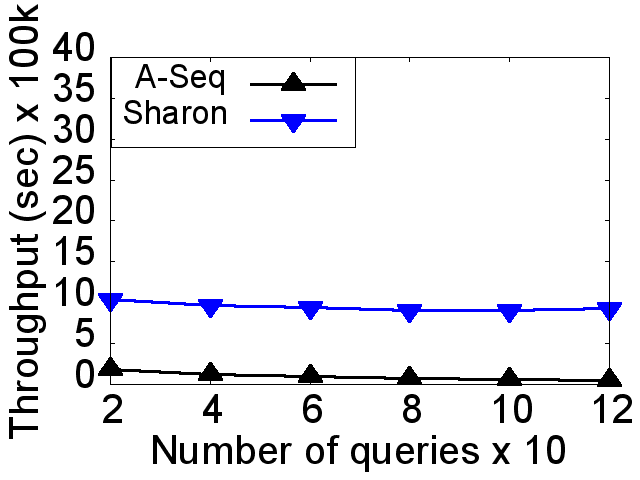}
    	\label{fig:thru-k}
	}
	\subfigure[Throughput (EC)]{
    	\includegraphics[width=0.23\columnwidth]{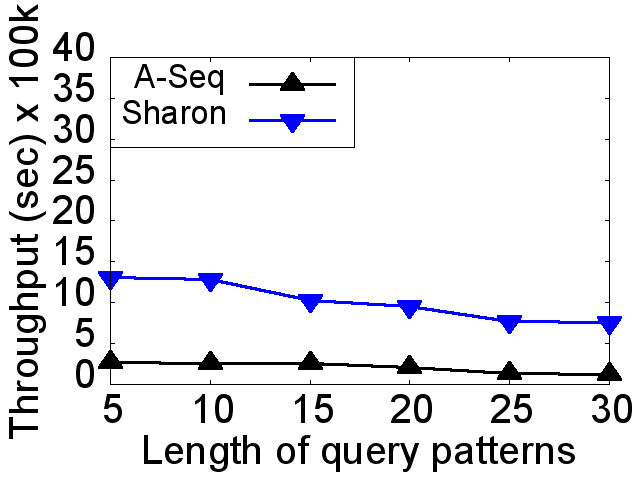}
    	 \label{fig:thru-l}
	}
	\subfigure[Memory (EC)]{
    	\includegraphics[width=0.23\columnwidth]{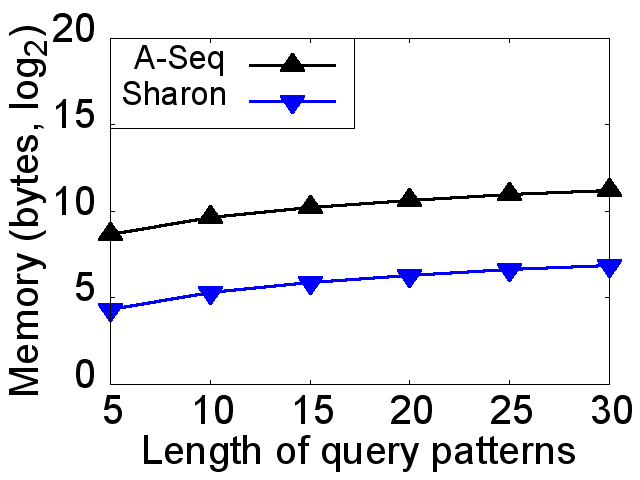}
    	\label{fig:mem-l}
	}    
	\caption{Online approaches (Taxi (TX), Linear Road (LR), and e-commerce (EC) data sets)}
	\label{fig:online}
\end{figure*}

\textbf{Online Approaches}. The online approaches (A-Seq and \app) perform similarly for such low-rate streams. They achieve five orders of magnitude speed-up compared to SPASS for 7k events per window because they aggregate event sequences without first constructing these sequences. 

Figure~\ref{fig:online} evaluates the online approaches against high-rate streams. We vary the number of events per window, the number of queries, and the length of their patterns and measure latency, throughput and memory of the online approaches.

The \textbf{\textit{Sharon Executor}} shares event sequences aggregation among all queries in the workload according to an optimal sharing plan that is computed based on an expanded \app\ graph (Section~\ref{sec:discussion}). 
The latency of \app\ and A-Seq grows linearly in the number of queries. \app\ achieves from 5--fold to 18--fold speed-up compared to \textbf{\textit{A-Seq}} when the number of queries increases from 20 to 120. Indeed, the more queries share their aggregation results, the fewer aggregates are maintained and the more events can be processed by the system (Figures~\ref{fig:latency-k} and \ref{fig:thru-k}).
\app\ requires up to two orders of magnitude less memory than A-Seq for 120 queries (Figure~\ref{fig:mem-k}). 
For low parameter values, \app\ defaults to A-Seq due to limited sharing opportunities.

While \app\ processes each event by each shared pattern exactly once,  each event can provoke repeated computations in A-Seq. Thus, the gain of \app\ grows linearly in the number of events per window. \app\ wins from 5--fold to 7--fold with respect to latency and throughput when the number of events increases from 200k to 1200k (Figures~\ref{fig:latency-r} and \ref{fig:thru-r}). 
Similarly, the speed-up of \app\ grows linearly from 4--fold to 6--fold with the increasing length of patterns (Figure~\ref{fig:latency-l}).
\app\ requires 20-fold less memory than A-Seq if the pattern length is 30 (Figure~\ref{fig:mem-l}).

Based on the experimental results in Figures~\ref{fig:two-steps} and \ref{fig:online}, we conclude that the latency, throughput and memory utilization of event sequence aggregation can be considerably reduced by the seamless integration of \textit{shared} and \textit{online} optimization techniques as proposed by our \app\ approach to enable real-time in-memory event sequence aggregation.

\begin{figure*}[t]
\centering
	\begin{minipage}{0.7\textwidth}	
	\subfigure[Latency]{
    	\includegraphics[width=0.48\textwidth]{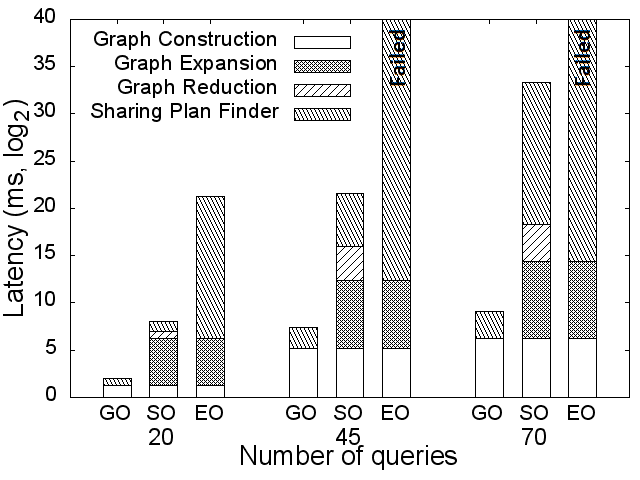}
    	 \label{fig:opt-latency}
	}
    \hspace*{2mm}
	\subfigure[Memory]{
    	\includegraphics[width=0.48\textwidth]{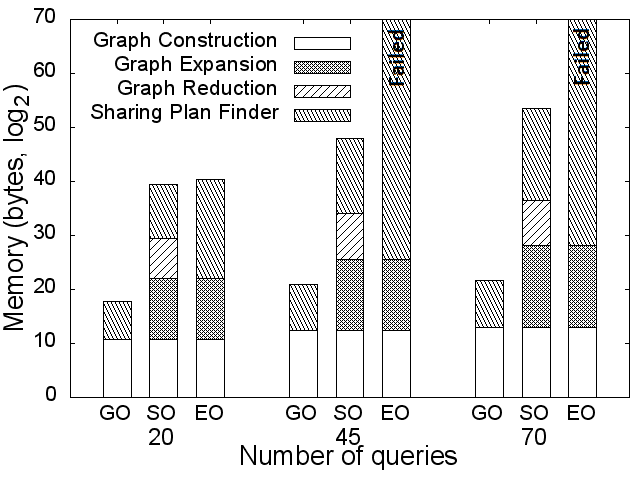}
    	\label{fig:opt-memory}
	}
	\caption{\app\ optimizer (SO) versus greedy optimizer (GO) and exhaustive optimizer (EO) (E-commerce query workload)}
	\label{fig:optim}
    \end{minipage}
    \hspace*{2mm}
    \begin{minipage}{0.25\textwidth}
    \centering
    \subfigure{
    	\includegraphics[width=0.85\textwidth]{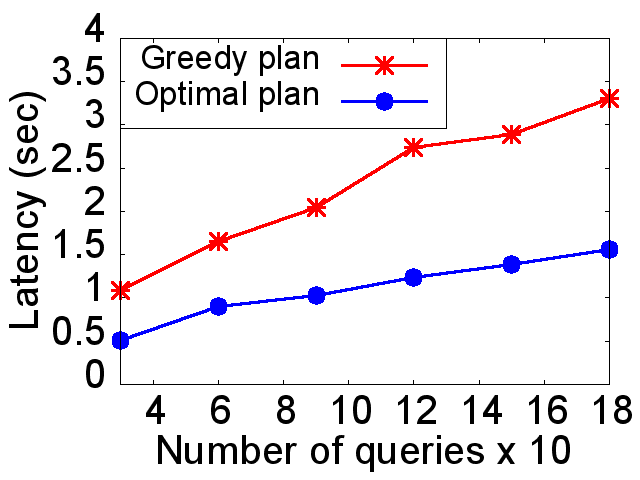}
    	\label{fig:2-plans-latency}
	}    
    \subfigure{
    	\includegraphics[width=0.85\columnwidth]{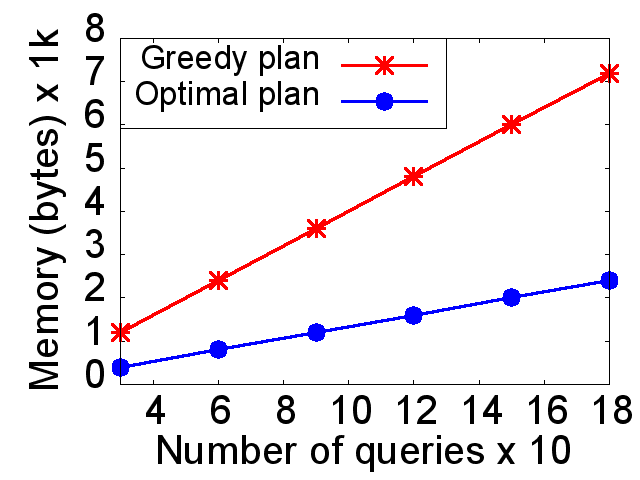}
    	\label{fig:2-plans-mem}
	}
	\caption{Sharing plan quality (Taxi data set)}
	\label{fig:2-plans}
    \end{minipage}
\vspace*{-4mm}
\end{figure*}

\subsection{Sharon Optimizer}
\label{exp:optimizer}

In Figure~\ref{fig:optim} we compare three optimizer solutions, while varying the number of queries. Each bar is segmented into phases as described below.

The \textbf{\textit{Greedy Optimizer}} consists of the following two phases: \app\ graph construction (Section~\ref{sec:graph}) and the GWMIN plan finder. 
In the worst case, both phases have polynomial latency and linear memory. However, our experiments show that on average more time and space is required to construct the \app\ graph than to run GWMIN. For 70 queries, 90\% of the time is spent constructing the graph.

The \textbf{\textit{Exhaustive Optimizer}} consists of three phases, namely, \app\ graph construction, graph expansion (Section~\ref{sec:discussion}), and an exhaustive search that traverses the entire search space to find an optimal plan. Thus, its latency and memory costs grow exponentially in the number of queries. The exhaustive optimizer fails to terminate for more than 20 queries. For 20 queries, its latency is 4 orders of magnitude higher than the latency of the greedy optimizer.

The \textbf{\textit{Sharon Optimizer}} consists of four phases, namely, \app\ graph construction, graph expansion, graph reduction, and the sharing plan finder that returns an optimal plan (Sections~\ref{sec:graph}--\ref{sec:discussion}). While its time and space complexity is exponential in the worst case (Equation~\ref{eq:complexity}), its latency and memory usage are reduced by our pruning principles compared to the exhaustive optimizer. 
On average, 36\% of the sharing candidates are pruned from the expanded \app\ graph, which is 99\% of the plan finder search space.
For 20 queries, \app\ outperforms the exhaustive optimizer by three orders of magnitude with respect to latency and by two orders of magnitude regarding memory usage. 

Our \app\ plan finder traverses the entire valid space to find an optimal plan. In contrast, GWMIN greedily selects one candidate with the highest benefit and eliminates its adjacent candidates from further consideration.
For example, for 70 queries, the latency of \app\ is three orders of magnitude higher, while its memory usage is two orders of magnitude larger compared to the greedy optimizer. 

\textbf{Sharing Plan Quality}. 
The greedy optimizer tends to return a sub-optimal sharing plan for two reasons.
One, it greedily selects a candidate $v$ with the maximal benefit in each step. By deciding to share $v$ it excludes all candidates adjacent to $v$ even though they may be more beneficial to share than $v$ alone. 
Two, the greedy optimizer does not resolve sharing conflicts  (Section~\ref{sec:discussion}). However, the sharing opportunities in the original \app\ graph may be rather limited (Figure~\ref{fig:graph}).  

In Figure~\ref{fig:2-plans}, we vary the number of queries and compare the latency and memory consumption of our \app\ executor when guided by a greedily chosen plan versus an optimal plan. We run these experiments on the Taxi real data set. The latency of the \app\ executor is reduced 2--fold and its memory consumption decreases 3--fold when 180 queries are processed according to an optimal plan compared to a greedily chosen plan. Thus, an optimal plan ensures real-time, light-weight event sequence aggregation.

\section{Related Work}
\label{sec:related_work}

\textbf{Complex Event Processing} (CEP) approaches such as SASE~\cite{ADGI08,ZDI14}, Cayuga~\cite{DGPRSW07}, and ZStream~\cite{MM09} support both event aggregation and event sequence detection over streams. 
SASE and Cayuga employ a Finite State Automaton (FSA)-based query execution paradigm, meaning that each event query is translated into an FSA. Each run of an FSA corresponds to a query match. 
In contrast, ZStream translates an event query into an operator tree that is optimized based on rewrite rules.
However, these approaches evaluate \textit{each query independently from other queries} in the workload -- causing both repeated computations and replicated storage in multi-query settings. 
Furthermore, they do not optimize event sequence aggregation queries -- which is the focus of our work. Thus, they require event sequence construction prior to their aggregation. Since the number of event sequences is polynomial in the number of events per window~\cite{ZDI14, QCRR14}, this \textit{two-step approach} introduces long delays for high-rate streams (Section~\ref{sec:evaluation}). 

In contrast, A-Seq~\cite{QCRR14} defines \textit{online} event sequence aggregation that eliminates event sequence construction. It incrementally maintains an aggregate for each pattern and discards an event once it updated the aggregates. We leverage this idea in our executor (Section~\ref{sec:benefit}). However, A-Seq has no optimizer to decide which patterns should be shared by which queries. Thus, A-Seq does not share event sequence aggregation. 
\greta~\cite{PLRM18} extends A-Seq by nested Kleene patterns and expressive predicates at the cost of storing of all matched events. Similarly to A-Seq, \greta\ optimizes \textit{single} queries.

\textit{CEP Multi-Query Optimization} (MQO) approaches such as SPASS~\cite{RLR16}, E-Cube~\cite{LRGGWAM11}, and RUMOR~\cite{hong2009rule} propose event sequence sharing techniques. 
SPASS exploits event correlation in an event sequence to determine the benefit of shared event sequence construction. 
E-Cube defines a concept and a pattern hierarchy of event sequence queries and  develops both top-down and bottom-up processing of patterns based on the results of other patterns in the hierarchy. 
RUMOR proposes a rule-based MQO framework for traditional RDBMS and stream processing systems. It defines a set of rules to merge NFAs representing different event queries. 
However, no optimization techniques for online aggregation of event sequences are proposed by the approaches above. They too construct all event sequences  prior to their aggregation. Event sequence construction degrades system performance.

\textbf{Data Streaming}. Streaming approaches typically support incremental aggregation~\cite{GSCL12, KWF06, AW04, LMTPT05, GHMAE07, LMTPT052, THSW15, ZKOS05, ZKOSZ10}. 
Some of them are shared. However, they solve an orthogonal problem. Namely, they enable shared aggregation given different windows, predicates, and group-by clauses~\cite{GSCL12, KWF06, AW04, LMTPT05}. Thus, they could be plugged into our approach as described in Section~\ref{sec:discussion}.
However, in contrast to \app, many of them aggregate only \textit{raw input events for single-stream queries}~\cite{GSCL12, LMTPT05, LMTPT052}.
Others evaluate simple Select-Project-Join queries with window semantics over data streams~\cite{KWF06}. They do not support CEP-specific operators such as event sequence that treat the order of events as a first-class citizen. Typically, they require the \textit{construction of join results} prior to their aggregation.

\textbf{Multi-Query Optimization} techniques include materialized views~\cite{MATViewICDE} and common sub-expression sharing~\cite{chakravarthy1986multiple,giannikis2010crescando} in relational databases. However, these approaches do not have the temporal aspect prevalent for CEP queries. Thus, they neither focus on event sequence computation nor their aggregation.
Furthermore, they assume that the data is statically stored on disk prior to processing. They neither target in-memory execution nor real-time responsiveness.

\section{Conclusions and Future Work}
\label{sec:conclusions}

Our \app\ approach is the first to enable \textit{shared online} event sequence aggregation. The \app\ optimizer encodes sharing candidates, their benefits and conflicts among them into the \app\ graph. Based on the graph, we define effective candidate pruning principles to reduce the search space of sharing plans. Our sharing plan finder returns an optimal plan to guide the executor at runtime. \app\ achieves an 18--fold speed-up compared to state-of-the-art approaches.

In the future, we plan to further investigate event sequence aggregation sharing for dynamic workloads to produce a new optimal sharing plan on the fly and migrate from the old to the new sharing plan with minimal overhead. 
Another interesting direction for future work is to leverage modern distributed multi-core clusters of machines to further improve the scalability of shared online event sequence aggregation.

\bibliographystyle{abbrv}
\bibliography{sharon-tr}

\section*{Appendix}

\begin{appendix}
\section{Sharable Pattern Detection}
\label{app:ccspan}

\begin{algorithm}[t]
\begin{algorithmic}[1]
\Require A query workload $Q$
\Ensure A hash table $S$ mapping a sharable pattern $p$ to a set of queries in $Q_p$ that contain $p$
\State $H, S \leftarrow \text{empty hash tables}$
\ForAll {$q$ in $Q$}
	\State $l \leftarrow q.\mathit{pattern}.\mathit{length}$
	\ForAll {$\mathit{end} = 0; \mathit{end} \leq l; \mathit{end}++$}
    \ForAll {$\mathit{start} = 0; \mathit{start} \leq \mathit{end}; \mathit{start}++$}
        \State $p = q.\mathit{pattern}.\mathit{substring}(\mathit{start},\mathit{end})$
        \If {$p.\mathit{length} > 1$} 
        \State $Q_p \leftarrow H.\mathit{get}(p); Q_p.\mathit{add}(q); H.\mathit{put}(p,Q_p)$ 
        \EndIf
    \EndFor        	
    \EndFor    
\EndFor
\ForAll {$p$ in $H$}
	\State $Q_p \leftarrow H.\mathit{get}(p)$
	\If {$Q_p.\mathit{size} > 1$} $S.\mathit{put}(p,Q_p)$
    \EndIf
\EndFor
\State \Return $S$
\end{algorithmic}
\caption{Modified CCSpan algorithm}
\label{ccspan-code}
\end{algorithm}

To detect sharable patterns (Definition~\ref{def:sharable}), we deploy a version of the CCSpan algorithm~\cite{ZWY15}. In this section, we first describe the original CCSpan algorithm. We then justify our changes and provide the modified algorithm. 

\textbf{Original CCSpan Algorithm}. 
CCSpan stands for \underline{C}losed \underline{C}ontiguous \underline{S}equential \underline{P}attern mining. A pattern is considered to be \textit{frequent} if it appears in more input sequences than the given support. A pattern is \textit{contiguous} if it is not interrupted by other patterns in the input sequences. Lastly, a pattern is called \textit{closed} if it cannot be further extended.

CCSpan adopts a pattern growth algorithm that records the pattern's occurrence in the input sequences. That is, starting from length 1, the algorithm recursively extends the patterns to their maximal length. CCSpan has linear time complexity in the number of input sequences assuming that the maximal length of input sequences is a small constant.

\textbf{Modified CCSpan Algorithm}.
Since shorter sequences can be shared between more queries than longer sequences, we detect not only frequent closed (or longest) sequences but also their sub-sequences. However, sharing a sequence of length one is not beneficial. Thus, we alter the original CCSpan algorithm to detect all \textit{frequent contiguous sequential patterns} of length $l > 1$. A pattern is considered to be \textit{frequent} if it appears in more than one query.

The modified CCSpan algorithm (Algorithm~\ref{ccspan-code}) consumes the query workload $Q$ and returns a hash table $S$ that maps each sharable pattern $p$ to the set of queries $Q_p \subseteq Q$ that contain $p$. Two empty hash tables $H$ and $S$ are initialized in Line~1. $S$ contains all sharable patterns, while $H$ maintains all patterns, i.e., $S \subseteq H$. For each query $q \in Q$, the algorithm considers each sub-pattern $p$ of the pattern of $q$ (Lines~3--6). If the length of $p$ is greater than one, the query $q$ is added to the set of queries $Q_p$ that $p$ appears within (Line~7--8). Lastly, we access each pattern $p$ in the hash table $H$ and if $p$ appears more than one query, the pattern $p$ and its respective queries $Q_p$ are added to the result $S$ (Lines~9--11). The hash table $S$ is returned in Line~12.

\textbf{Complexity Analysis}.
Let $n$ be the number of queries in $Q$ and $l$ be the maximal length of a pattern. Then, the three for-loops in Lines~2--8 are called $O(nl^2)$ times. In addition, the for-loop in Lines~9--11 iterates $O(nl)$ times. All operations on the hash tables $H$ and $S$ happen in constant time. In summary, the time complexity is linear in $n$, i.e., $O(nl^2) + O(nl) = O(n)$ since $l$ is a small constant in practice.
The space complexity is determined by the size of the hash tables $H$ and $S$. The number of stored patterns is $O(nl)$. Each pattern is mapped to $O(n)$ queries. In summary, the space complexity is quadratic in $n$, i.e., $O(n^2 l) = O(n^2)$.

\section{GWMIN Algorithm}
\label{app:gwmin}

\begin{algorithm}[t]
\begin{algorithmic}[1]
\Require A weighted graph $G=(V,E)$ 
\Ensure An independent set $\mathit{IS}$ 
\State $\mathit{IS}\leftarrow\emptyset;\ i\leftarrow 0;\ G_i\leftarrow G$
\While {$V(G_i) \neq \emptyset$}
	\State $\mathit{max}\leftarrow0$
		\ForAll {$v$ in $V(G_i)$}
        	\State $\mathit{new}\_\mathit{max} \leftarrow \frac{\mathit{weight}(v)}{\mathit{degree}_{G_i}(v)+1}$
    		\If {$\mathit{new}\_\mathit{max} > \mathit{max}$}
            	\State $\mathit{max} \leftarrow \mathit{new}\_\mathit{max};\ v_i \leftarrow v$
    		\EndIf
    	\EndFor
    \State $\mathit{IS} \leftarrow \mathit{IS} \cup \{v_i\};\ G_{i+1} \leftarrow G_i[V(G_i) \setminus \mathcal{N}_{G_i}^+(v_i)];\ i++$
\EndWhile
\State \Return $\mathit{IS}$
\end{algorithmic}
\caption{GWMIN algorithm}
\label{gwmin-code}
\end{algorithm}

The GWMIN algorithm finds the MWIS in a graph~\cite{Sakai2003}. GWMIN stands for \underline{G}reedy \underline{Min}imum degree algorithm for \underline{W}eighted graphs. 
Algorithm~\ref{gwmin-code} consumes a weighted graph $G=(V,E)$ and returns its independent set $\mathit{IS}$. At the beginning, the set $\mathit{IS}$ is empty, the iteration counter $i=0$, and the graph in $i^{\mathit{th}}$ iteration $G_i = G$ (Line~1). In each iteration, the algorithm selects the vertex $v$ with the maximal ratio
\begin{equation*}
\frac{\mathit{weight}(v)}{\mathit{degree}_{G_i}(v)+1}
\end{equation*}
where $\mathit{weight}(v)$ is the weight and $\mathit{degree}_{G_i}(v)$ is the degree of $v$ in i$^{\mathit{th}}$ iteration (Lines~3--7). This vertex $v$ is added to the independent set $\mathit{IS}$, $v$ and its neighbors are removed from the graph $G$ (Line~8). Once the graph $G$ is empty, the algorithm returns the independent set $\mathit{IS}$ (Lines~2, 9).

\textbf{Complexity Analysis}. 
The for-loop in Lines~4--7 is called $O(|V|)$ times. The time complexity of removing a vertex $v$ and its neighbors from the graph is $O(|E|)$. Thus, the time complexity of one iteration of the while-loop in Lines 2--9 is $O(|V|+|E|)$. This while-loop is called $O(|V|)$ times. In summary, the time complexity is $O(|V| (|V| + |E|)$.
The space complexity is linear in the size of the graph $G$ and its independent set $\mathit{IS}$, i.e., $O(|V|+|E|)$.
\end{appendix}

\end{document}